\documentclass[journal]{IEEEtran}
\usepackage{subfigure}
\usepackage{amsmath,color}
\usepackage{graphicx,amssymb,lineno,bm,subfigure}
\usepackage{algorithm}
\usepackage{algorithmic}
\usepackage{cite}
\newtheorem{coro}{\textbf{Corollary}}
\newtheorem{rmk}{\textbf{Remark}}

\newtheorem{prop}{\textbf{Proposition}}
\newtheorem{proof}{\textbf{Proof}}
\newtheorem{lma}{\textbf{Lemma}}

\begin{document}

\title{\huge{Grid Energy Consumption and QoS Tradeoff in Hybrid Energy Supply Wireless Networks}}
\author{\normalsize \authorblockN{Yuyi~Mao,~\IEEEmembership{Student Member,~IEEE},
Jun~Zhang,~\IEEEmembership{Senior Member,~IEEE}, and Khaled~B. Letaief,~\IEEEmembership{Fellow,~IEEE}}

\thanks{Manuscript received April 1, 2015; revised August 4, 2015 and
November 11, 2015; accepted January 16, 2016. The associate editor
coordinating the review of this paper and approving it for
publication was Dr. Federico Boccardi. This work was supported by the Hong Kong
Research Grants Council under Grant No. 610212.
%The associate editor coordinating the review of this paper and approving it for publication was Dr. Andres Kwasinski.
}
\thanks{The authors are with the Department of Electronic and Computer Engineering, the Hong Kong University
of Science and Technology, Clear Water Bay, Kowloon, Hong Kong (e-mail: \{ymaoac, eejzhang, eekhaled@ust.hk\}). Khaled B. Letaief is also with Hamad bin Khalifa University, Doha, Qatar (e-mail: kletaief@hkbu.edu.qa).}

\thanks{This work was presented in part at IEEE Wireless Communications and Networking Conference (WCNC), New Orleans, LA, Mar. 2015.}
%\IEEEauthorblockA{Dept. of ECE, The Hong Kong University of Science and Technology\\
%Email: \{ymaoac, eejzhang, eekhaled\}@ust.hk}
}

%\IEEEauthorblockA{Dept. of ECE, The Hong Kong University of Science and Technology\\
%Email: \{ymaoac, eejzhang, eekhaled\}@ust.hk}

%\markboth{Draft}%
%{Shell \MakeLowercase{\textit{et al.}}: Bare Demo of IEEEtran.cls for Journals}

\maketitle

\begin{abstract}
%\boldmath
Hybrid energy supply (HES) wireless networks have recently emerged as a new paradigm to enable green networks, which are powered by both the electric grid and harvested renewable energy. In this paper, we will investigate two critical but conflicting design objectives of HES networks, i.e., the grid energy consumption and quality of service (QoS). Minimizing grid energy consumption by utilizing the harvested energy will make the network environmentally friendly, but the achievable QoS may be degraded due to the intermittent nature of energy harvesting. To investigate the tradeoff between these two aspects, we introduce the total service cost as the performance metric, which is the weighted sum of the grid energy cost and the QoS degradation cost. Base station assignment and power control is adopted as the main strategy to minimize the total service cost, while both cases with non-causal and causal side information are considered. With non-causal side information, a Greedy Assignment algorithm with low complexity and near-optimal performance is proposed. With causal side information, the design problem is formulated as a discrete Markov decision problem. Interesting solution structures are derived, which shall help to develop an efficient monotone backward induction algorithm. To further reduce complexity, a Look-Ahead policy and a Threshold-based Heuristic policy are also proposed. Simulation results shall validate the effectiveness of the proposed algorithms and demonstrate the unique grid energy consumption and QoS tradeoff in HES networks.
\end{abstract}

\begin{keywords}
Green communications, energy harvesting (EH), hybrid energy supply (HES), power control, base station assignment, Markov decision process (MDP), QoS.
\end{keywords}
% \fi
\IEEEpeerreviewmaketitle

\section{Introduction}
\subsection{Related Works and Motivations}

\IEEEPARstart{W}ITH the explosion of mobile data traffic in wireless ecosystems, the energy consumption of the Information and Communication Technology (ICT) industry is becoming an economic and ecological issue. It is estimated that the annual carbon emission and electric grid power consumption of the ICT industry will reach up to 235 Mto \cite{Fehske11} and 414 TWh \cite{Lambert12}, respectively, in 2020. Thus, it is crucial to seek alternative environmentally friendly energy sources for wireless networks. Energy harvesting (EH) is a promising candidate, which can capture ambient recyclable energy, including solar radiation, radio signal, wind energy, etc \cite{Sudevalayam11}, and thus it is becoming an indispensable solution for green communication systems \cite{Ulukus15}. By introducing EH techniques to the next-generation communication systems, it is estimated that a 20\% $CO_{2}$ emission reduction can be achieved \cite{Piro13}.

Due to the intermittent nature of the EH processes, conventional communication protocols cannot take full advantages of the harvested energy. Consequently, transmission policies for communication systems with EH capabilities have been widely investigated. The optimal transmission policy for point-to-point EH systems with off-line side information (SI), including channel SI (CSI) and energy SI (ESI), was studied in \cite{Ozel11}, which was shown to have a \emph{directional water-filling} structure. The investigation has been later extended to scenarios with more practical assumptions on SI \cite{ZWang14,Blasco13,YLuo1212}, and to other systems, including broadcast channels \cite{JYang12} and multiple access channels  \cite{JYang122}, as well as cooperative communication systems \cite{CHuang13,YLuo13,YLuo1312,YLuo1601,YMao14}. Moreover, joint energy-bandwidth allocation policies for multi-user EH systems have been recently proposed in \cite{ZWangPartI,ZWangPartII}.

Because of the innate characteristic of EH systems, it is difficult to maintain stable communication performance. To get full benefits of the renewable energy sources while guaranteeing the quality of service (QoS) requirement, the electric grid can be retained as a backup energy source. As a result, recently, hybrid energy supply (HES) communication systems have attracted significant attentions, where EH and electric grid coexist. The initial investigation started from point-to-point HES systems. Recursive expressions for the power grid energy consumption in HES systems were derived
in \cite{YMao1312}. Given the grid energy budget, a \emph{two-stage water-filling} algorithm was proposed for throughput maximization in \cite{JGong13}, where the CSI and ESI was assumed to be non-causally known at the transmitter. In \cite{XKang1408}, power allocation strategies for weighted energy cost minimization were proposed. For HES \emph{orthogonal frequency division multiple access} (OFDMA) systems, a power and sub-carrier allocation  policy to maximize energy-efficiency was developed in \cite{DNg13}. More recently, there have been substantial efforts on HES wireless networks. Targeting on saving the grid energy consumption, the green energy utilization optimization problem was solved in \cite{THan1302}. In \cite{JGong14}, joint optimization of sleep control and resource allocation for HES networks was investigated. For HES heterogeneous networks, a green energy and latency aware user association scheme was proposed in \cite{THan13} to minimize the grid energy consumption while achieving the required QoS. Besides, in \emph{coordinated multi-point} (CoMP) HES systems, joint energy cooperation and communication cooperation was proposed in \cite{JXu14}. In addition, a low-complexity online base station assignment and power control algorithm based on Lyapunov optimization was developed in \cite{YMao1512}.

Most of the existing works on HES systems assume that both the harvested energy and grid energy are available at the same transmitter. However, considering practical implementation issues, e.g. the high deployment cost of the power cable and the increasing number of base stations (BSs), it is difficult to connect all the BSs to the electric grid \cite{YMao1501}. Consequently, in future HES networks, there will be a portion of BSs powered purely by the EH devices, i.e., without electric cable connection \cite{YMao15}.  Motivated by the greenness brought by EH techniques, together with the constraints on the power cable deployment, in this paper, we will investigate how to coordinate a grid-powered BS (GP-BS) and an EH-BS to serve a mobile user. The number of transmitted data packets is used as a measure of the achievable QoS. In such networks, minimizing the grid energy consumption at the GP-BS and maximizing the QoS are two important but conflicting targets. That is, in order to achieve a better QoS, more grid energy will be consumed, and vice versa. Consequently, to balance the grid energy consumption and the degradation of QoS, i.e., the number of dropped data packets, it is of fundamental importance to investigate their tradeoff.

Different from previous studies with both grid energy and harvested energy at the same BS, new design challenges arise in the HES network considered in this paper. In particular, both BS assignment, i.e., whether the GP-BS or the EH-BS should be selected to serve the mobile user, and BS power control, i.e., how to determine the transmit power of the picked BS in each block, need to be optimized. Compared to conventional wireless networks, where users are normally assigned to their nearest BSs \cite{CLi14}, the major design challenge in HES wireless networks comes from the sporadic energy arrival at the EH-BS. Thus, the BS assignment should depend on both the channel condition and the EH profile. In addition, since the EH-BS does not have a grid energy supply, efficient utilization of the harvested energy should adapt to conditions of multiple channels from both the EH-BS and GP-BS, which complicates the design compared to systems with co-located grid energy and harvested energy.

\subsection{Contributions}
In this paper, we will study the grid energy consumption and QoS tradeoff in HES wireless networks. Our major contributions are summarized as follows:
\begin{enumerate}
\item
{We consider an HES wireless network with a GP-BS and an EH-BS, which is motivated by practical implementation considerations. Such a network model has not been considered in previous works on HES systems, which only assumed a single BS with co-located grid energy and harvested energy \cite{JGong13,XKang1408,DNg13,THan1302,JGong14,THan13,JXu14}.}

\item
{Due to the difference in the network model, the basic design principle is different from previous studies. In particular, we propose a novel performance metric, namely, the \emph{total service cost}, to optimize the HES network. This new metric helps to reveal a new design aspect of the HES network, i.e., the grid energy consumption and QoS tradeoff. This is innovative compared to existing works on HES systems, which aimed at minimizing the grid energy consumption under given QoS requirements.}

\item
{We investigate the total service cost minimization (TSCM) problem assuming either non-causal or causal side information, which provides a viable approach to balance the grid energy consumption and the achievable QoS, i.e., the number of dropped packets. Specifically:}
\begin{itemize}
\item
{With non-causal CSI and ESI, i.e., the off-line setting, the TSCM problem turns out to be a highly difficult mixed integer nonlinear programming problem. By exploiting the problem structure, we transform the TSCM problem into an equivalent zero-one integer programming problem, which, however, is still NP-hard. Thus, a low-complexity Greedy Assignment algorithm is proposed, which is suboptimal in general but achieves a close-to-optimal performance. For two special cases, we show that the greedy algorithm is actually optimal. This not only offers key intuitions on developing the online algorithms with causal SI, but also facilitates the evaluation of the online algorithms without the need of calling for high-complexity solvers.}

\item
{With causal CSI and ESI, i.e., the online setting, we formulate the TSCM problem into a discrete Markov decision process (MDP) problem. Interesting monotone structures of the optimal policy are derived, which coincide with the intuitions obtained from the Greedy Assignment algorithm and help to develop an efficient monotone backward induction algorithm. In order to further reduce complexity, a Look-Ahead policy and a Threshold-based Heuristic policy are also proposed.}
\end{itemize}

\item
Extensive simulation results are provided to validate the effectiveness of the proposed policies. In particular, we will show that the proposed Greedy Assignment algorithm achieves a near-optimal performance for the off-line case, while the Threshold-based Heuristic policy can approach the discrete MDP solution for the online case. Moreover, the fundamental grid energy consumption and QoS tradeoff in such networks is illustrated. Based on the tradeoff curves, for a given QoS requirement, the operating point to minimize the grid energy consumption can be determined via adjusting the weighting factor. Our results reveal that intelligent energy management schemes that fully exploit the available SI play an essential role in HES wireless networks.
\end{enumerate}

\subsection{Organization}
The rest of this paper is organized as follows. We introduce the system model in Section II. In Section III,  the TSCM problem with non-causal SI is formulated, for which an efficient algorithm is developed. In Section IV, the case with causal SI is studied. Simulation results are presented in Section V and we will conclude this paper in Section VI.

\section{System Model}

\subsection{Network Model and Energy Model}
\begin{figure}[h]
\begin{center}
    \label{System Model Fig}
   \includegraphics[width=0.48\textwidth]{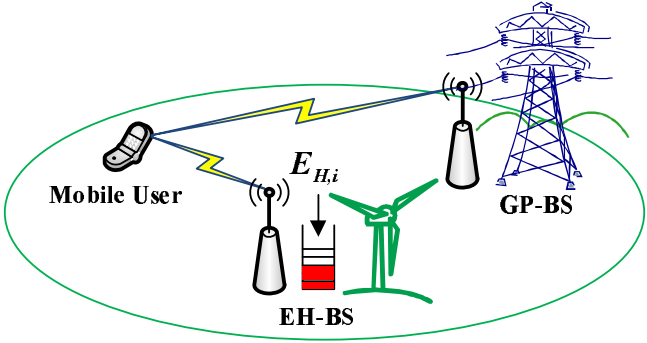}
\end{center}
\caption{An HES network with an EH-BS, a GP-BS and a mobile user.}
\label{systemmodel}
\end{figure}
As shown in Fig. \ref{systemmodel}, we consider an HES wireless network with an EH-BS, a GP-BS and a mobile user. The GP-BS, which can be a macro-BS to maintain network coverage, is powered by the electric grid and its maximum transmit power is denoted as $p_{G}^{\max}$. The EH-BS is equipped with an EH component and powered purely by the harvested energy, with the maximum transmit power as $p_{H}^{\max}$. This EH-BS may serve as a small cell to increase system capacity \cite{YMao15}. These two BSs will coordinate to serve the considered user. Such deployment of HES networks is practically attractive for occasions where the electric grid is not easily accessible. We consider an $N$-block downlink transmission frame with block length $\tau$, where the user can be served by at most one BS in each block\footnote{In this paper, we will adopt BS assignment to optimize the network, i.e., the user is served by one of the BSs at each time, which is a widely-used transmission scheme for cellular networks \cite{CLi14,Jeff11}. In principle, joint transmission schemes, i.e., the BSs transmit simultaneously, can be used for HES networks \cite{JZhang0904}. However, it will cause significant feedback overhead and have more stringent requirement on synchronization \cite{Biermann13,QCui14}.}. The block index set is denoted as $\mathcal{N}=\{1,2,...,N\}$.

We assume a packet with $R$ bits arriving at the beginning of each block, and it may be transmitted via the EH-BS or the GP-BS in the following block. It may also happen that neither of the BSs will serve the user. The number of transmitted packets is used to measure the achievable QoS \cite{AGiovanidis09}, while BS assignment and power control is adopted to optimize the system. The GP-BS will collect all the required side information, and make the transmission decision. In this paper, the delay incurred by signalling and decision making is neglected for simplicity. Denote $I_{j,i}\in \{0,1\}$, with $j\in\{G, H, D\}$, as the BS assignment indicators in the $i$-th block, where $I_{G,i}=1$ ($I_{H,i}=1$) indicates that the GP-BS (EH-BS) is assigned to deliver the packet, and $I_{D,i}=1$ represents the event that neither of the BSs is selected, i.e., the packet is dropped. These indicators are subject to the following operation constraint:
\begin{eqnarray}
I_{G,i}+I_{H,i}+I_{D,i}=1,\forall i \in \mathcal{N}.
\label{operation_constraint}
\end{eqnarray}
In addition, the transmit power of the GP-BS and EH-BS in the $i$-th block is denoted as $p_{G,i}$ and $p_{H,i}$, respectively. In this work, energy consumed for purposes other than transmission, e.g. cooling and baseband signal processing, is neglected.

The EH process is modeled as successive energy packet arrivals, i.e., at the beginning of the $i$-th block, an energy packet with $E_{H,i}$ Joule arrives at the EH-BS and is available for transmission from then on. Without loss of generality, we assume $E_{H,i}$'s are independent and identically distributed (i.i.d.) within $\left[0,E_{m}\right]$ with probability density function (PDF) $f_{E_{H}}\left(e\right)$ and average harvesting power $\mathrm{P_{avg}}$. A battery with sufficiently large capacity is installed at the EH-BS. Since the energy packets cannot be used before their arrivals, the following energy causality constraint should be satisfied\footnote{One interesting extension is to incorporate the wireless backhaul energy consumption at the EH-BS, which can be handled if the backhaul energy consumption at each block is known. In this case, `$p_{H,l}$' in (\ref{EH constraint}) can be regarded as the total power consumption of both the downlink transmission and wireless backhaul in the $l$-th block.}:
\begin{eqnarray}
\sum_{l=1}^{i}p_{H,l}\tau\leq \sum_{n=1}^{i}E_{H,n},\forall i\in\mathcal{N}.
\label{EH constraint}
\end{eqnarray}

Both path loss and small-scale fading are considered. The channels are assumed to be i.i.d. block fading. Denote $\gamma_{G,i}$ and $\gamma_{H,i}$ as the small-scale fading channel gains of the GP-BS to user channel (called G-channel) and the EH-BS to user channel (called H-channel) in the $i$-th block, respectively. Thus, the channel gains of the G-channel and H-channel are given by
$h_{G,i}=g_{0}d_{G}^{-\theta}\gamma_{G,i}$ and $h_{H,i}=g_{0}d_{H}^{-\theta}\gamma_{H,i}$, respectively,
where $g_{0}$ is the path loss constant, $\theta$ is the path loss exponent, $d_{G}$ and $d_{H}$ are the distances from the GP-BS to the user and from the EH-BS to the user, respectively. In the $i$-th block, when the GP-BS is assigned to serve the user, the amount of bits that can be transmitted is given by $r\left(p_{G,i},h_{G,i}\right)$. Similarly, when the EH-BS is selected, $r\left(p_{H,i},h_{H,i}\right)$ bits can be delivered. Here,
$r\left(p,h\right)=\tau W\log_{2}\left(1+h\sigma^{-2}p\right)$ is the Shannon-Hartley formula,
where $W$ is the system bandwidth and $\sigma^{2}$ is the noise variance at the receiver.

\subsection{Performance Metric}
In most previous studies on HES communication systems, minimizing the grid energy consumption while guaranteeing certain QoS requirement is considered as a fundamental design objective \cite{JGong13,XKang1408,THan13,IAhmed1312}. However, with practical considerations as mentioned before, it may happen that some data packets could not be transmitted, i.e., the QoS requirement cannot be satisfied. For instance, the EH-BS may not have enough energy while the channel from the GP-BS is in deep fading. In this case, neither of the BSs is capable of delivering the data packet. In some real-time applications, this packet may indeed be dropped, while for other applications, this will increase the transmission delay. To take this aspect into consideration, we assume that a unit of cost will be induced once a packet is dropped. On the other hand, due to the intermittent and sporadic nature of the EH process of the EH-BS, the GP-BS needs to serve the user from time to time, which incurs one unit of cost of grid energy consumption per Joule. Hence, it is desirable to minimize both costs from the system design perspective, which, unfortunately, is impossible in general. Thus, the tradeoff between these conflicting targets is of particular interest. To characterize this tradeoff, we introduce the total service cost over the $N$-block transmission frame as the performance metric, which is the weighted sum of the grid energy cost and the packet drop cost, with weights $w_{G}$ and $w_{D}$, respectively. Specifically, the total service cost is defined as
\begin{equation}
\mathrm{TSC} \triangleq \sum_{i=1}^{N} w_{G} p_{G,i}\tau+w_{D} I_{D,i},
\end{equation}
where $w_{G} p_{G,i} \tau$ and $w_{D} I_{D,i}$ are the weighted grid energy cost and the packet drop cost in the $i$-th block, respectively. It is worthwhile to note that, although the costs of per packet drop and per Joule of grid energy consumption are assumed to be normalized in this paper, we can also handle non-normalized values by embodying them into the weights $w_{D}$ and $w_{G}$, respectively.

\begin{rmk}
By adjusting the weight of the grid energy cost, $w_{G}$, and the weight of the packet drop cost, $w_{D}$, we can achieve different grid energy consumption and QoS tradeoffs in such networks. More specifically, when $w_{G} \gg w_{D}$, the network is grid energy sensitive. On the other hand, when $w_{D} \gg w_{G}$, the network emphasizes more on the successful packet delivery, i.e., addresses more on QoS.
\end{rmk}

In the following two sections, we will investigate the TSCM problems under non-causal and causal SI settings, which will reveal the tradeoff between the grid energy consumption and the achievable QoS.

\section{Total Service Cost Minimization with Non-causal Side Information}
In this section, we will first solve the off-line TSCM problem, i.e., assuming $\{E_{H,i}\}$, $\{h_{G,i}\}$ and $\{h_{H,i}\}$ are non-causally known. This will provide insights for designing such systems, as well as serving as a performance upper bound for the online case.

Denote $\mathbf{I}_{i}\triangleq\left[I_{G,i},I_{H,i},I_{D,i}\right]$. The off-line TSCM problem is then formulated as
\begin{align}
\mathcal{P}_{1}: &\mathop{\min}_{\mathbf{I}_{i},p_{j,i}} \sum_{n=1}^N w_{G} p_{G,n}\tau +w_{D} I_{D,n}\\
&\ \mathrm{s.t.}\ (\ref{operation_constraint}) \ \mathrm{and}\ (\ref{EH constraint})\\
&\ \ \ \sum_{j\in\{G,H\}}I_{j,i}r\left(p_{j,i},h_{j,i}\right)\geq \left(1-I_{D,i}\right)R, \forall i \in \mathcal{N}\label{QoS}\\
&\ \ \ \ \ \ p_{j,i}\leq p_{j}^{\max}, \forall i \in \mathcal{N}, j\in\{G,H\}\label{PPC}\\
&\ \ \ \ \ \ I_{j,i} \in\{0,1\}, \forall i\in \mathcal{N},j\in\{G,H,D\}. \label{zeroone}
\end{align}
In $\mathcal{P}_{1}$, the objective is the total service cost over the $N$-block frame. (\ref{QoS}) is the QoS constraint, i.e., when it is decided to transmit the packet in the $i$-th block, the channel capacity should be greater than the packet size. (\ref{PPC}) and (\ref{zeroone}) are the peak power constraint and the zero-one indicator constraint, respectively.

In general, $\mathcal{P}_{1}$ is a mixed integer nonlinear programming (MINLP) problem with continuous variables $p_{j,i}$ together with binary variables $\mathbf{I}_{i}$, which is very difficult to solve. By exploiting the problem structure, we manage to transform $\mathcal{P}_{1}$ into an integer programming (IP) problem, which is a special form of the \emph{multi-dimensional Knapsack problem} \cite{multiKnapsack}, as shown in Lemma 1.
\begin{lma}
$\mathcal{P}_{1}$ can be equivalently transformed into the following IP problem:
\begin{align}
\hat{\mathcal{P}}_{1}: &\mathop{\min}_{\alpha_{i}} \sum_{n=1}^N \left(1-\alpha_{n}\right)c_{n}\label{objP4}\\
&\ \mathrm{s.t.}  \sum_{k=1}^{i} \alpha_{k}p_{H,k}^{inv}\tau \leq \sum_{l=1}^{i} E_{H,l}, \forall i \in \mathcal{N}\label{EH causal constraint3}\\
&\ \ \ \ \ \alpha_{i}p_{H,i}^{inv}\leq p_{H}^{\max}, \forall i \in \mathcal{N}\label{H PPC2}\\
&\ \ \  \ \ \alpha_{i}\in\{0,1\}, \forall i\in \mathcal{N}, \label{zeroone2}
\end{align}
where $p_{H,i}^{inv}$ is the H-channel inversion power, i.e., $p_{H,i}^{inv}=
\left(2^{\frac{R}{W \tau}}-1\right)\sigma^{2}h_{H,i}^{-1}$. $\{c_{i}\}$ is a series of cost parameters given by
\begin{eqnarray}
c_{i}=
\begin{cases}
w_{D}, & p_{G,i}^{inv}> \min\{p_{G}^{\max},w_{D} \left(w_{G} \tau \right)^{-1}\}\\
w_{G}  p_{G,i}^{inv} \tau, & p_{G,i}^{inv}\leq \min\{p_{G}^{\max},w_{D} \left(w_{G} \tau \right)^{-1}\},
\end{cases}
\label{costpara}
\end{eqnarray}
where $p_{G,i}^{inv}$ is the G-channel inversion power, i.e., $p_{G,i}^{inv}=\left(2^{\frac{R}{W\tau}}-1\right)\sigma^{2}h_{G,i}^{-1}$. Given a solution $\{\alpha_{i}\}$ to $\hat{\mathcal{P}}_{1}$, the corresponding solution to $\mathcal{P}_{1}$ is given as
\begin{equation}
\begin{split}
I_{H,i}&=\alpha_{i}\\
I_{G,i}&=\left(1-\alpha_{i}\right)\bm{1}\{p_{G,i}^{inv}\leq \kappa \}\\
I_{D,i}&=1-I_{G,i}-I_{H,i}\\
p_{j,i}&=I_{j,i}p_{j,i}^{inv}, j\in\{G,H\},
\end{split}
\label{soluP3}
\end{equation}
where $\kappa \triangleq \min\{p_{G}^{\max},w_{D}\left(w_{G} \tau\right)^{-1}\}$ and $\bm{1}\{\cdot\}$ is the indicator function. Furthermore, if $\{\alpha_{i}\}$ is optimal for $\hat{\mathcal{P}}_{1}$, $\{\mathbf{I}_{i}\}$, $\{p_{H,i}\}$ and $\{p_{G,i}\}$ are optimal for $\mathcal{P}_{1}$.
\label{lmaSecIII1}
\end{lma}

\begin{proof}
See Appendix A.
\end{proof}
Lemma 1 shows that transmitting the packet with the channel inversion power will not affect the optimality. To solve $\hat{\mathcal{P}}_{1}$, we only need to decide when to assign the EH-BS, i.e., $\{\alpha_{i}\}$. If the EH-BS is not assigned to serve the user in a block, whether to drop the packet or to transmit it with the GP-BS only depends on the relationship between $p_{G,i}^{inv}$, $p_{G}^{\max}$, $w_{G}$ and $w_{D}$. These problem structures help reduce the size of the optimization variables. However, $\hat{\mathcal{P}}_{1}$ is still an NP-hard problem, as shown in the following corollary.
\begin{coro}
$\hat{\mathcal{P}}_{1}$ is an NP-hard problem.
\end{coro}
\begin{proof}
Suppose there is a polynomial-time algorithm to solve $\hat{\mathcal{P}}_{1}$, then there exists a polynomial-time algorithm to solve the following \emph{Knapsack problem}:
\begin{align}
&\mathop{\min}_{y_{i}} \sum_{n=1}^N \left(1-y_{n}\right)d_{n}\label{knapsack}\nonumber\\
&\ \mathrm{s.t.}  \sum_{k=1}^{N}  y_{k}f_{k} \leq e\nonumber\\
&\ \ \  \ \ y_{i}\in\{0,1\}, \forall i\in \mathcal{N}, \nonumber
\end{align}
with $e,d_{i},f_{i}\geq 0,\forall i\in\mathcal{N}$, which is a special case of $\hat{\mathcal{P}}_{1}$ when $E_{H,i}=0, \forall i\in\mathcal{N}\setminus \{1\}$. Since the \emph{Knapsack problem} is NP-hard \cite{KellerKnapsack}, i.e., there is no polynomial-time algorithm for the optimal solution, the assumption is violated. Therefore, $\hat{\mathcal{P}}_{1}$ is also an NP-hard problem.
\end{proof}

Though commercial solvers, e.g. MOSEK \cite{MOSEK}, can be applied to compute the optimal solution for $\hat{\mathcal{P}}_{1}$, they are known to have a worst-case exponential complexity and cannot provide insights on the obtained solution. Thus, low-complexity algorithms are needed. To assist the algorithm design, we first provide the following property of the optimal solution.
\begin{prop}
If $\{\alpha_{i}\}$ is optimal to $\hat{\mathcal{P}}_{1}$, then there exists no $j>i$, such that $\alpha_{i}=1$, $\alpha_{j}=0$, $c_{i}< c_{j}$ and $p_{H,i}^{inv}\geq p_{H,j}^{inv}$.
\label{SecIIIprop1}
\end{prop}
\begin{proof}
Suppose $\{\alpha_{i}\}$ is optimal and there exists $j>i$, such that $\alpha_{i}=1$, $\alpha_{j}=0$, $c_{i}< c_{j}$ and $p_{H,i}^{inv}\geq p_{H,j}^{inv}$. It is always possible to construct a new solution by setting $\alpha_{i}=0$ and $\alpha_{j}=1$. With the constructed solution, the value of the objective function will decrease by $c_{j}-c_{i}$. Hence, we can conclude that $\{\alpha_{i}\}$ is not optimal, which contradicts the assumption.
\end{proof}

We call the block in which the EH-BS is assigned to serve the user as an H-block, i.e., $\alpha_{i}=1$. From Proposition \ref{SecIIIprop1}, the H-blocks should have higher values of $c_{i}$ and $h_{H,i}$. The intuition is that, we can reduce more service cost with a given amount of harvested energy. In the following, we will propose the Greedy Assignment (GA) algorithm to solve $\hat{\mathcal{P}}_{1}$, which identifies the H-blocks greedily according to a metric function $\Xi\left(p,q\right): \mathbb{R}_{+}^{2} \rightarrow \mathbb{R}_{+}$, which is non-decreasing with $p$ for a given $q$, and non-increasing with $q$
for a given $p$. With the assistance of this metric function, the solution obtained by the GA algorithm will satisfy Proposition \ref{SecIIIprop1} automatically. The GA algorithm is summarized in Algorithm \ref{GAalgorithm}.

\begin{algorithm}[h]
\caption{Greedy Assignment Algorithm.}
\begin{algorithmic}[1]
\REQUIRE
{$\mathbf{E_{H}}=\{E_{H,i}\}$, $\bm{c}=\{c_{i}\}$, $\{p_{H,i}^{inv}\}$, $p_{H}^{\max}$, $\tau$} and $N$
\ENSURE
{$\bm{\alpha}=\{\alpha_{i}\}$, and $c_{\Sigma}^{\small\mathrm{GA}}$}
\STATE Initialize $\bm{\alpha}=\mathbf{0}$, $\mathcal{N}_{\mathrm{H}^{c}}=\{1,2,...,N\}$, $\mathcal{N}_{\mathrm{feas}}=\emptyset$;
\STATE $\left[I,\mathcal{N}_\mathrm{feas}\right]=\mathrm{find\_feas}\left(\mathbf{E_{H}},\bm{\alpha},\{p_{H,i}^{inv}\},p_{H}^{\max},\tau,\mathcal{N}_{\mathrm{H}^{c}}\right)$;
\STATE \textbf{While} {$I==1$} \textbf{do}
\STATE \hspace{10pt}pick $m$, such that $m=\arg\max\limits_{i\in\mathcal{N}_{\mathrm{feas}}} \Xi\left(c_{i},p_{H,i}^{inv}\right)$;
\STATE \hspace{10pt}$\bm{\alpha}=\bm{\alpha}+\bm{e}_{m}$;
\STATE \hspace{10pt}$\mathcal{N}_{\mathrm{H}^{c}}=\mathcal{N}_{\mathrm{H}^{c}}\setminus\{m\}$;
%\STATE $\mathbf{\hat{E}_{H}}=\mathrm{update\_profile}\left(\mathbf{\hat{E}_{H}},m\right)$;
\STATE \hspace{10pt}$\left[I,\mathcal{N}_\mathrm{feas}\right]=\mathrm{find\_feas}\left(\mathbf{E_{H}},\bm{\alpha},\{p_{H,i}^{inv}\},p_{H}^{\max},\tau,\mathcal{N}_{\mathrm{H}^{c}}\right)$;
\STATE \textbf{End while}
\STATE $c_{\Sigma}^{\small\mathrm{GA}}=\left(\mathbf{1}-\bm{\alpha}\right)^{\mathrm{T}}\bm{c}$;
\STATE Return $\bm{\alpha}=\{\alpha_{i}\}$, $c_{\Sigma}^{\small\mathrm{GA}}$.
\end{algorithmic}
\label{GAalgorithm}
\end{algorithm}

In Algorithm \ref{GAalgorithm}\footnote{$\bm{y}=\{y_{i}\}$ represents an $N\times 1$ column vector, $\bm{e}_{m}$ is a basic vector with the $m$-th entry equal to 1, and $\emptyset$ is the empty set.}, we initialize $\bm{\alpha}$ by setting all blocks not to be H-blocks. $\mathcal{N}_{\mathrm{H}^{c}}$ is the set of blocks that have not yet been selected as H-blocks. The function $\mathrm{find\_feas}\left(\cdot\right)$ finds the block from $\mathcal{N}_{\mathrm{H}^{c}}$, which is feasible if it is selected as an H-block given the current $\bm{\alpha}$, and returns the feasible block index set $\mathcal{N}_{\mathrm{feas}}$ together with an indicator $I$, which can be explicitly written as
\begin{equation}
\begin{split}
&\mathcal{N}_{\mathrm{feas}}=\\
&\bigg\{i\bigg|i\in\mathcal{N}_{\mathrm{H}^{c}},p_{H,i}^{inv}{\leq} p_{H}^{\max}, \sum_{k=1}^{n}\alpha_{k}'p_{H,k}^{inv}\tau{\leq} \sum_{l=1}^{n}E_{H,l}, \forall n\in\mathcal{N}\bigg\},
\end{split}
\end{equation}
where $\bm{\alpha}'=\bm{\alpha}+\bm{e}_{i}$ and $I=\bm{1}\{\mathcal{N}_{\mathrm{feas}} \neq \emptyset\}$. In each loop, the block from $\mathcal{N}_{\mathrm{feas}}$ with the highest value of $\Xi\left(c_{i},p_{H,i}^{inv}\right)$ is selected as the H-block. Key properties of the GA algorithm are listed as follows:
\begin{enumerate}
\item
As an H-block is selected in each while loop, the GA algorithm will be terminated within at most $N$ loops. In each loop, the block with the maximum value of the metric function will be identified as an H-block and the blocks that are still feasible will be found. Thus, the complexity of the GA algorithm is $\mathcal{O}\left(N^{2}\right)$.
\item
From extensive simulation results provided in Section V, the GA algorithm achieves near-optimal performance.
\item
When $h_{H,i}=h_{H}$ or $h_{G,i}=h_{G},\forall i\in\mathcal{N}$, i.e., when one of the channels is constant, the GA algorithm is optimal. The proof is given in Appendix B.\footnote{Although the proposed low-complexity GA algorithm is only optimal in some special cases, i.e., it is not guaranteed to be a performance upper bound, its close-to-optimal performance (which can be observed from the simulation results in Section V) demonstrates its capability to benchmark the performance achieved by the online algorithms without the need of calling for high-complexity commercial solvers.}
\end{enumerate}

The off-line formulation requires full SI, which is too restrictive and thus impractical. In the next section, we will turn our attention to the online case, where SI is only causally known. As will be seen latter, the intuition obtained from the GA algorithm, i.e., the EH-BS should serve the user when the H-channel condition is good while the G-channel condition is poor, matches the structural properties of the optimal online solution and helps to develop a low-complexity suboptimal online policy.

\section{Total Service Cost Minimization with causal Side Information}
In this section, we focus on the online case, where the EH profile and channel states are causally known. Consequently, we are interested in minimizing the expected total service cost. Thus, the online TSCM problem can be formulated as
\begin{align}
\mathcal{P}_{2}: &\mathop{\min}_{\bm{\pi}\in \Pi} \mathbb{E}^{\bm{\pi}}\left[\sum_{n=1}^N \left(1-\alpha_{n}\right)c_{n}\right]\label{objP2},
\end{align}
where $\bm{\pi}$ denotes a feasible policy (which will be specified later), $\Pi$ represents the set of all feasible policies, $c_{n}$ is defined in (\ref{costpara}) and the expectation is over all the randomness involved. In principle, the optimal solution of this problem can be obtained by finite-horizon dynamic programming (DP) algorithms, which, however, suffer from the \emph{curse of dimensionality}. To avoid the high computational complexity of the DP solutions, we will reformulate the online optimization problem as a discrete Markov decision process (MDP) problem, where the system is described by a finite number of states in contrast to DP.

\subsection{A Discrete MDP Approach}
Denote the battery capacity as $B_{m}$. The MDP formulation consists of the following components:

1) \emph{System state:}
The system state in the $i$-th block, denoted as $\bm{s}_{i}$, is represented by a triplet, i.e.,
\begin{equation}
\bm{s}_{i}\triangleq \langle \epsilon_{i} , \gamma_{G,i} , \gamma_{H,i}\rangle,
\label{stateMDP}
\end{equation}
 where $\epsilon_{i}$ is the battery energy state. We partition the battery capacity into $M$ intervals with equal lengths, i.e.,
$\left[0,B_{m}\right]=\left[0,B_{m}/M\right)\cup\ldots\cup\left[\left(M-1\right)B_{m}/{M},B_{m}\right]$, and each of them is an energy state represented by its mid-value, i.e., $\epsilon_{i}\in\{B_{m}/{2M},\ldots,\left(2M-1\right)B_{m}/{2M}\}$. The small-scale fading channel gain is quantized into $K$ non-overlapping intervals
using the equi-probable steady state method \cite{HSWang95}, and each interval is a channel state represented by the corresponding mean value $H_{k}$, i.e., $\gamma_{G,i}, \gamma_{H,i} \in\{H_{1},\ldots,H_{K}\}$ and $\mathrm{Pr}\{\gamma_{j,i}=H_{k}\}=1/K$.

2) \emph{Action space:}
The action in the MDP problem is $\alpha_{i}$, and the action space is $\{0,1\}$. The allowable action space $\mathcal{A}_{\bm{s}_{i}}$, is the feasible action set in the state $\bm{s}_{i}$, i.e.,
\begin{equation}
\mathcal{A}_{\bm{s}_{i}}=
\begin{cases}
\{0\}, &p_{H,i}^{inv}> \min\{\epsilon_{i}\slash\tau,p_{H}^{\max}\}\\
\{0,1\}, &{\rm{otherwise}},
\end{cases}
\end{equation}
which indicates that the EH-BS can transmit only when the harvested energy is sufficient and the peak power constraint is not violated as well.

3) \emph{State transition probability:}
The state transition probability, denoted as $p\left(\bm{s}_{i+1}|\bm{s}_{i},\alpha_{i}\right)$, is the probability that the system will be in state $\bm{s}_{i+1}$ in the $\left(i+1\right)$-th block, given that in the $i$-th block it was in state $\bm{s}_{i}$ and action $\alpha_{i}$ was taken.
Since the i.i.d. block fading channel is quantized into equal probability states, we have
\begin{equation}
\begin{split}
&p\left(\bm{s}_{i+1}|\bm{s}_{i},\alpha_{i}\right)\\&=\mathrm{Pr}\{\gamma_{G,i+1}=H_{m}\}\cdot \mathrm{Pr}\{\gamma_{H,i+1}=H_{n}\}\cdot p\left(\epsilon_{i+1}|\bm{s}_{i},\alpha_{i}\right)\\
&=K^{-2}\cdot p\left(\epsilon_{i+1}|\bm{s}_{i},\alpha_{i}\right), \forall m,n=1,\cdots,K,
\end{split}
\end{equation}
where $p\left(\epsilon_{i+1}|\bm{s}_{i},\alpha_{i}\right)$ is the energy state transition probability. The energy state evolves according to $\epsilon_{i+1}=Q\left(\epsilon_{i}-\alpha_{i}p_{H,i}^{inv}\tau+E_{H,i+1}\right)$, where
\begin{equation}
\begin{split}
&Q\left(\varepsilon\right)\\&=\left(2\min\bigg\{\bigg\lfloor\frac{M \min\{\varepsilon,B_{m}\}}{B_{m}}\bigg \rfloor+1,M\bigg\}-1\right)\cdot\frac{B_{m}}{2M},
\end{split}
\end{equation}
and $\lfloor x \rfloor$ is the floor function.  $Q\left(\varepsilon\right)$ is a quantization function mapping $\varepsilon$ to the battery energy states, which is non-decreasing over $\varepsilon$.

4) \emph{Cost function, decision rule and policy:} The cost function is $c\left(\bm{s}_{i},\alpha_{i}\right)=\left(1-\alpha_{i}\right)c_{i}$, where $c_{i}$ depends on $\bm{s}_{i}$ as given by (\ref{costpara}). The decision rule is a mapping from the state space to the action space, i.e., $d_{i}: \bm{s}_{i}\rightarrow \alpha_{i}$. A policy $\bm{\pi}$ is a sequence of decision rules for each block, i.e., $\bm{\pi}=\{d_{1}^{\bm{\pi}}\left(\bm{s}_{1}\right),...,d_{N}^{\bm{\pi}}\left(\bm{s}_{N}\right)\}$. The set of all admissible policy is denoted as $\Pi$. If the state in the first block is $\bm{s}_{1}$, the expected total service cost can be written as
\begin{eqnarray}
c_{\Sigma}^{\bm{\pi}}\left(\bm{s}_{1}\right)=
\mathbb{E}^{\bm{\pi}}_{\bm{s}}\left[{\sum_{n=1}^{N}c\left(\bm{s}_{n},d_{n}^{\bm{\pi}}\left(\bm{s}_{n}\right)\right)}\bigg|\bm{s}_{1}\right],
\end{eqnarray}
where the expectation is over all possible state sequences $\bm{s}=\{\bm{s}_{1},...,\bm{s}_{N}\}$ induced by $\bm{\pi}$. The optimal policy and the optimal expected total service cost can be expressed as
\begin{equation}
\bm{\pi}^{*}=\arg\mathop{\min}_{\bm{\pi}\in\Pi} c_{\Sigma}^{\bm{\pi}}\left(\bm{s}_{1}\right)\ \mathrm{and}\  c^{*}_{\Sigma}\left(\bm{s}_{1}\right)=\mathop{\min}_{\bm{\pi}\in\Pi}c_{\Sigma}^{\bm{\pi}}\left(\bm{s_{1}}\right),
\end{equation}
respectively. The cost-to-go function in the $i$-th block, which is the sum of the expected total service cost from the $i$-th block to the last block, is given by
\begin{equation}
\begin{split}
&u_{i}^{\bm{\pi}}\left(\bm{s}_{i}\right)=\\
&\begin{cases}
c\left(\bm{s}_{i},d_{i}^{\bm{\pi}}\left(\bm{s}_{i}\right)\right)+
\sum\limits_{\bm{s}'}p\left(\bm{s}'|\bm{s}_{i},d_{i}^{\bm{\pi}}\left(\bm{s}_{i}\right)\right)
u^{\bm{\pi}}_{i+1}\left(\bm{s}'\right),&i<N \\
c\left(\bm{s}_{N},d_{N}^{\bm{\pi}}\left(\bm{s}_{N}\right)\right),&i=N.
\end{cases}
\end{split}
\end{equation}
According to the Principle of Optimality, the optimal cost-to-go functions should satisfy the following Bellman's equations \cite{BertsekasDP}:
\begin{equation}
\begin{split}
&u_{i}^{*}\left(\bm{s}_{i}\right)=\\
&\begin{cases}
\min\limits_{\alpha_{i}\in \mathcal{A}_{\bm{s}_{i}}}
\bigg\{c\left(\bm{s}_{i},\alpha_{i}\right)+\sum\limits_{\bm{s}'} p\left(\bm{s}'|\bm{s}_{i},\alpha_{i}\right)u_{i+1}^{*}\left(\bm{s}'\right)\bigg\},&i<N\\
\min\limits_{\alpha_{N}\in \mathcal{A}_{\bm{s}_{N}}}c\left(\bm{s}_{N},\alpha_{N}\right),&i=N.
\end{cases}
\label{Bellman}
\end{split}
\end{equation}
The optimal policy $\bm{\pi}^{*}$ of this MDP problem can be obtained by solving (\ref{Bellman}) recursively with the backward induction algorithm (BIA) \cite{BertsekasDP}, in which, the cost-to-go functions are evaluated over all possible states in each block. Next, we will derive interesting monotone structures of $\bm{\pi}^{*}$, which will help us design computationally efficient algorithms by reducing such evaluations.
Firstly, $\bm{\pi}^{*}$ is monotone over the small-scale fading of the G-channel, as specified in Proposition \ref{propSecIVA1}.

%%%%%%%%%%%%%%%%
\begin{prop}
For a given battery energy state $\epsilon_{i}$, if the H-channel is in state $\gamma_{H,i}$, then the optimal policy is monotone over $\gamma_{G,i}$. Specifically, the EH-BS will be assigned to serve the user, i.e., $\alpha_{i}=1$, only when $\gamma_{G,i}\leq \Gamma_{G,i}\left(\epsilon_{i},\gamma_{H,i}\right)$, where $\Gamma_{G,i}\left(\epsilon_{i},\gamma_{H,i}\right)$ is the threshold G-channel state.
\label{propSecIVA1}
\end{prop}

\begin{proof}
Suppose $\gamma_{G,i}^{+}\geq \gamma_{G,i}^{-}$, and denote $\bm{s}_{i}^{+}=\langle\epsilon_{i},\gamma_{G,i}^{+},\gamma_{H,i}\rangle$ and
$\bm{s}_{i}^{-}=\langle\epsilon_{i},\gamma_{G,i}^{-},\gamma_{H,i}\rangle$. To prove Proposition \ref{propSecIVA1}, it suffices to show if $d_{i}^{*}\left(\bm{s}_{i}^{+}\right)=1$, then $d_{i}^{*}\left(\bm{s}_{i}^{-}\right)=1$. For $i=N$, this result is straightforward. For $i<N$, since $d_{i}^{*}\left(\bm{s}_{i}^{+}\right)=1$, according to (\ref{Bellman}),
\begin{equation}
\begin{split}
&c\left(\bm{s}_{i}^{+},1\right)+\sum_{\bm{s}'}p\left(\bm{s}'|\bm{s}_{i}^{+},1\right)u^{*}_{i+1}\left(\bm{s}'\right)
\\&\leq c\left(\bm{s}_{i}^{+},0\right)+\sum_{\bm{s}'}p\left(\bm{s}'|\bm{s}_{i}^{+},0\right)u^{*}_{i+1}\left(\bm{s}'\right).
\end{split}
\end{equation}
Since $c\left(\bm{s}_{i}^{+},1\right)=c\left(\bm{s}_{i}^{-},1\right)=0$, $c\left(\bm{s}_{i}^{+},0\right)\leq c\left(\bm{s}_{i}^{-},0\right)$ and $p\left(\bm{s}|\bm{s}_{i}^{+},\alpha\right)=p\left(\bm{s}|\bm{s}_{i}^{-},\alpha\right),\forall \alpha, \bm{s}$,
\begin{equation}
\begin{split}
&c\left(\bm{s}_{i}^{-},1\right)+\sum\limits_{\bm{s}'}p\left(\bm{s}'|\bm{s}_{i}^{-},1\right)u^{*}_{i+1}\left(\bm{s}'\right)
\\&\leq c\left(\bm{s}_{i}^{-},0\right)+\sum\limits_{\bm{s}'}p\left(\bm{s}'|\bm{s}_{i}^{-},0\right)u^{*}_{i+1}\left(\bm{s}'\right),
\end{split}
\end{equation}
i.e., $d^{*}_{i}\left(\bm{s}_{i}^{-}\right)=1$, which ends the proof.
\end{proof}

The following two lemmas will help to derive the monotone structure of $\bm{\pi}^{*}$ over the H-channel state $\gamma_{H,i}$.
\begin{lma}
Suppose $f\left(\epsilon\right)$ is a non-increasing function of the energy state $\epsilon$. If $0\leq u_{1}\leq u_{2}\leq \epsilon_{0}$, then $\sum_{\epsilon_{1}}p\left(\epsilon_{1}|\epsilon_{0},u_{1}\right)f\left(\epsilon_{1}\right)\leq
\sum_{\epsilon_{2}}p\left(\epsilon_{2}|\epsilon_{0},u_{2}\right)f\left(\epsilon_{2}\right)$, where $p\left(\epsilon_{n}|\epsilon_{0},u_{n}\right)$ is the probability that the battery will be in energy state $\epsilon_{n}$ in the next block given the current energy state $\epsilon_{0}$ and the amount of harvested energy ($u_{n}$) consumed in the current block ($n=1,2$).
\label{lmaSecIVA1}
\end{lma}
\begin{proof}
The proof can be obtained by rewriting $\sum_{\epsilon_{n}}p\left(\epsilon_{n}|\epsilon_{0},u_{n}\right)f\left(\epsilon_{n}\right)$ as
\begin{equation}
\int_{e}f_{E_{H}}\left(e\right)f\left(Q\left(\epsilon_{0}-u_{n}+e\right)\right)de
\end{equation}
for $n=1,2$ and utilizing the non-decreasing property of $Q\left(\varepsilon\right)$.
\end{proof}

Define $\hat{u}_{i}^{*}\left(\epsilon\right)\triangleq\sum\limits_{\gamma_{G,i},\gamma_{H,i}}u^{*}_{i}\left(\langle\epsilon
,\gamma_{G,i},\gamma_{H,i}\rangle\right)$, which can be regarded as the normalized expected total service cost from the $i$-th block given the battery is in energy state $\epsilon$, and the future expected total service cost can be represented in terms of $\hat{u}_{i+1}^{*}\left(\epsilon\right)$ as
\begin{equation}
\begin{split}
&\ \ \ \sum_{\bm{s}'}p\left(\bm{s}'|\bm{s}_{i},\alpha_{i}\right)u_{i+1}^{*}\left(\bm{s}'\right)\\
&=K^{-2}\sum_{\epsilon'}p\left(\epsilon'|\bm{s}_{i},\alpha_{i}\right)\sum_{\gamma_{G}',\gamma_{H}'}u^{*}_{i+1}\left(\langle \epsilon',\gamma_{G}',\gamma_{H}'\rangle\right)\\
&=K^{-2}\sum_{\epsilon'}
p\left(\epsilon'|\bm{s}_{i},\alpha_{i}\right)\hat{u}_{i+1}^{*}\left(\epsilon'\right).
\end{split}
\end{equation}
The following lemma shall demonstrate that a higher energy state leads to a lower expected total service cost.

\begin{lma}
$\forall i\in\mathcal{N}$, $\hat{u}_{i}^{*}\left(\epsilon\right)$ is non-increasing with $\epsilon$.
\label{lmaSecIVA2}
\end{lma}
\begin{proof}
See Appendix C.
\end{proof}

With the assistance of Lemma \ref{lmaSecIVA1} and Lemma \ref{lmaSecIVA2}, we can obtain the monotone structure of  $\bm{\pi}^{*}$ over the H-channel state $\gamma_{H,i}$ as shown in the following proposition.

\begin{prop}
For a given battery energy state $\epsilon_{i}$, if the G-channel is in state $\gamma_{G,i}$, then the optimal policy is monotone over $\gamma_{H,i}$. Specifically, the EH-BS will be assigned to serve the user, i.e., $\alpha_{i}=1$, only when $\gamma_{H,i}\geq \Gamma_{H,i}\left(\epsilon_{i},\gamma_{G,i}\right)$, where $\Gamma_{H,i}\left(\epsilon_{i},\gamma_{G,i}\right)$ is the threshold H-channel state.
\label{propSecIVA2}
\end{prop}
\begin{proof}
Suppose $\gamma_{H,i}^{+}\geq \gamma_{H,i}^{-}$, and denote $\bm{s}_{i}^{+}=\langle\epsilon_{i},\gamma_{G,i},\gamma_{H,i}^{+}\rangle$ and $\bm{s}_{i}^{-}=\langle\epsilon_{i},\gamma_{G,i},\gamma_{H,i}^{-}\rangle$. To prove Proposition 3, it suffices to show if $d_{i}^{*}\left(\bm{s}_{i}^{-}\right)=1$, then $d_{i}^{*}\left(\bm{s}_{i}^{+}\right)=1$. For $i=N$, the result is straightforward. For $i<N$, since $d_{i}^{*}\left(\bm{s}_{i}^{-}\right)=1$, according to (\ref{Bellman}),
\begin{equation}
\begin{split}
&c\left(\bm{s}_{i}^{-},1\right)+
K^{-2}\sum\limits_{\epsilon'}p\left(\epsilon'|\bm{s}_{i}^{-},1\right)\hat{u}^{*}_{i+1}\left(\epsilon'\right)\\
&\leq c\left(\bm{s}_{i}^{-},0\right)+
K^{-2}\sum\limits_{\epsilon'}p\left(\epsilon'|\bm{s}_{i}^{-},0\right)\hat{u}^{*}_{i+1}\left(\epsilon'\right) \\
& \overset{(a)}{=}
c\left(\bm{s}_{i}^{+},0\right)+
K^{-2}\sum\limits_{\epsilon'}p\left(\epsilon'|\bm{s}_{i}^{+},0\right)\hat{u}^{*}_{i+1}\left(\epsilon'\right).
\end{split}
\end{equation}
($a$) is because the cost parameter $c_{i}$ is irrelevant with $\gamma_{H,i}$ and there is no harvested energy consumption when $\alpha_{i}=0$. By combining Lemma \ref{lmaSecIVA1} and Lemma \ref{lmaSecIVA2}, we have
\begin{equation}
\sum_{\epsilon'}p\left(\epsilon'|\bm{s}_{i}^{+},1\right)\hat{u}_{i+1}^{*}\left(\epsilon'\right)
\leq \sum_{\epsilon'}p\left(\epsilon'|\bm{s}_{i}^{-},1\right)\hat{u}_{i+1}^{*}\left(\epsilon'\right),
\end{equation}
since EH-BS transmits with lower power when $\gamma_{H,i}$ is higher. Hence,
\begin{equation}
\begin{split}
&c\left(\bm{s}_{i}^{+},1\right)+
K^{-2}\sum\limits_{\epsilon'}p\left(\epsilon'|\bm{s}_{i}^{+},1\right)\hat{u}^{*}_{i+1}\left(\epsilon'\right)\\
&\leq c\left(\bm{s}_{i}^{+},0\right)+
K^{-2}\sum\limits_{\epsilon'}p\left(\epsilon'|\bm{s}_{i}^{+},0\right)\hat{u}^{*}_{i+1}\left(\epsilon'\right)
,
\end{split}
\end{equation}
i.e., $d_{i}^{*}\left(\bm{s}_{i}^{+}\right)=1$, which ends the proof.
\end{proof}

According to Proposition \ref{propSecIVA1} and Proposition \ref{propSecIVA2}, if $d_{i}^{*}\left(\langle \epsilon_{i},\gamma_{G,i},\gamma_{H,i}\rangle\right)=1$, then  $d_{i}^{*}\left(\langle \epsilon_{i},\gamma_{G,i}',\gamma_{H,i}'\rangle\right)$ $=1,\forall \gamma_{G,i}'\leq \gamma_{G,i},\gamma_{H,i}'\geq \gamma_{H,i}$. These propositions also indicate that, it would be more beneficial to assign the EH-BS to serve the user when the H-channel is in a good condition while the G-channel is poor. Interestingly, this conclusion coincides with the intuitions we found while developing the GA algorithm. That is, we can save more service cost with less harvested energy if the H-blocks are identified sequentially according to the metric function $\Xi\left(c_{i},p_{H,i}^{inv}\right)$, which uncovers the inherent characteristics of HES wireless networks.

In order to develop the algorithm for the optimal policy, we also need the following properties of the optimal cost-to-go function.
\begin{prop}
$\forall i\in\mathcal{N}$, $u^{*}_{i}\left(\bm{s}_{i}\right)$ obeys the following properties:
\begin{enumerate}
\item $u^{*}_{i}\left(\langle \epsilon_{i},\gamma_{G,i},\gamma_{H,i}\rangle\right)=u_{i}^{*}\left(\langle \epsilon_{i},\Gamma_{G,i}\left(\epsilon_{i},\gamma_{H,i}\right) ,\gamma_{H,i}\rangle\right)$, $\forall \gamma_{G,i}\leq \Gamma_{G,i}\left(\epsilon_{i},\gamma_{H,i}\right) $.

\item $u^{*}_{i}\left(\langle \epsilon_{i},\gamma_{G,i},\gamma_{H,i}\rangle\right)=u_{i}^{*}\left(\langle \epsilon_{i},\gamma_{G,i},\tilde{\Gamma}_{H,i}\left(\epsilon_{i},\gamma_{G,i}\right) \rangle\right)$, $\forall \gamma_{H,i}\leq \tilde{\Gamma}_{H,i}\left(\epsilon_{i},\gamma_{G,i}\right) $, where $\tilde{\Gamma}_{H,i}\left(\epsilon_{i},\gamma_{G,i}\right)$ $\triangleq \max\{H_{k}|\mathrm{s. t.} H_{k}< \Gamma_{H,i}\left(\epsilon_{i},\gamma_{G,i}\right)\}$.
\end{enumerate}
\label{propSecIVA3}
\end{prop}
\begin{proof}
The proof can be obtained with the monotone structures derived in Proposition \ref{propSecIVA1} and Proposition \ref{propSecIVA2}.
Details are omitted for brevity.\end{proof}

By exploiting the monotone structures, we propose a monotone backward induction algorithm (MBIA) with the main idea illustrated via an example in Fig. \ref{Monotone Illustration}. Given the block index $i$ and the battery energy state $\epsilon_{i}$, assume that we traverse the $K^{2}$ channel states $\langle\gamma_{G,i},\gamma_{H,i}\rangle$ with an order as $\langle H_{5}, H_{5}\rangle$, $\langle H_{4}, H_{5}\rangle$,\ldots,$\langle H_{1},H_{5}\rangle$, $\langle H_{5}, H_{4}\rangle$, $\langle H_{4}, H_{4}\rangle$,\ldots,$\langle H_{2}, H_{1}\rangle$, $\langle H_{1}, H_{1}\rangle$. The BIA evaluates the cost-to-go function over all 25 channel states as shown in Fig. \ref{BIA}. In the MBIA, as shown in Fig. \ref{MBIA} after evaluating the cost-to-go function in state $\langle H_{5} , H_{5}\rangle$, in which the optimal decision is $\alpha_{i}^{*}=1$,  there is no need to go through the states with a lower G-channel gain, i.e., the next evaluation will be in state $\langle H_{5} , H_{4}\rangle$. Since $\alpha^{*}_{i}=0$ in state $\langle H_{5} , H_{4}\rangle$, i.e., $\Gamma_{H,i} \left(\epsilon_{i},H_{5}\right)=H_{5}> H_{4}$, there is no need to go through the states with a lower H-channel gain. Thus, after the evaluation at $\langle H_{4} , H_{4}\rangle$, the MBIA will go to $\langle H_{4} , H_{3}\rangle $ directly. Such procedures carry on until the optimal actions in all states are determined, and eventually, only 9 evaluations of the cost-to-go function are needed.

Details of the MBIA are given in Algorithm \ref{MBIAalgorithm}, where the major differences compared to the conventional BIA lie on the procedure of checking the threshold states, i.e., from Line 9 to 16. Based on Propositions \ref{propSecIVA1}, \ref{propSecIVA2} and \ref{propSecIVA3}, once the threshold states are found, the optimal action for the states with a lower G-channel (H-channel)  gain will be determined as $1$ ($0$). Meanwhile, the  value of the optimal cost-to-go function $u^{*}_{i}\left(\bm{s}_{i}\right)$ in these remaining states will be assigned the same as that of their corresponding threshold states. Once $\bm{\pi}^{*}$ is obtained, it acts as a look-up table. During data transmission, based on the system state, the operation can be determined immediately by referring to the associated entry in such a table.

\begin{figure}[H]
\centering
\subfigure[BIA]{
\label{BIA}
\includegraphics[width=0.23\textwidth]{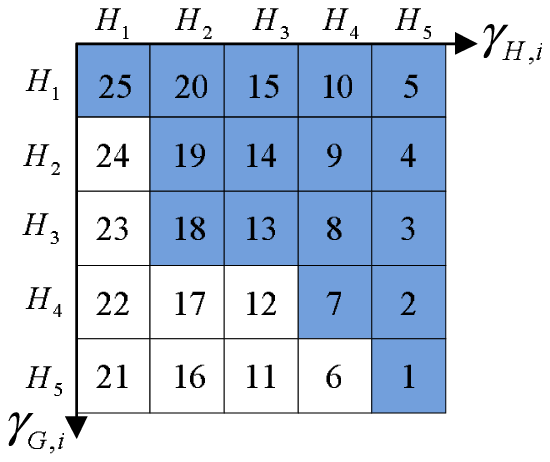}}
\subfigure[MBIA]{
\label{MBIA}
\includegraphics[width=0.23\textwidth]{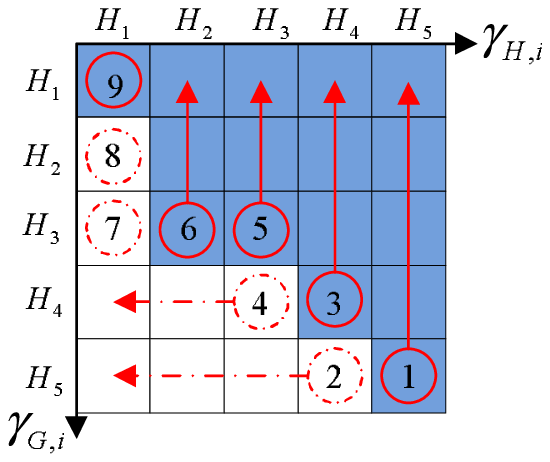}}
\caption{Illustration of the monotone structures, $K=5$ (The colored states are with $d_{i}^{*}\left(\bm{s}_{i}\right)=1$. The number in each state is the traverse order. The states with solid-line (dot dash-line) circles are the threshold states with $d_{i}^{*}\left(\bm{s}_{i}\right)=1$ ($d_{i}^{*}\left(\bm{s}_{i}\right)=0$)).}
\label{Monotone Illustration}
\end{figure}

\begin{rmk}
Although we have discovered the monotone structures of  $\bm{\pi}^{*}$ over the channel gains, unfortunately, such structures do not hold over the battery energy level $\epsilon_{i}$. In other words, as $M$ increases, the complexity of the MBIA may still be high. Consequently, it is important to choose a suitable resolution for the battery energy state to balance the computational complexity and the achievable performance.
\end{rmk}
\begin{algorithm}[h]
\caption{Monotone Backward Induction Algorithm.}
\begin{algorithmic}[1]
\REQUIRE
{$B_{m}$, $w_{G}$, $w_{D}$, $g_{0}$, $\theta$, $\sigma^{2}$, $d_{H}$, $d_{G}$, $\{H_{k}\}$, $p_{G}^{\max}$, $p_{H}^{\max}$, $\tau$,  $M$, $K$} and $N$
\ENSURE
{$\bm{\pi}_{N\times M\times K \times K}$}
\STATE Initialize $\mathbf{u}^{*}=+\infty \times \bm{1}_{N\times M\times K \times K}, \bm{\pi}=\mathbf{0}_{N\times M\times K \times K}$, $t=N$;
\STATE \textbf{While} {$t\geq 1$} \textbf{do}
\STATE \hspace{10pt} \textbf{For} {$m=1$ to $M$} \textbf{do}
\STATE \hspace{20pt} $k_{1}=K,k_{2}=K$;
\STATE \hspace{20pt} \textbf{While} {$k_{1}\geq 1 $ and $k_{2}\geq 1 $} \textbf{do}
\STATE \hspace{30pt} $\bm{s}=\langle \left(2m-1\right)B_{m}/2M,H_{k_{2}},H_{k_{1}}\rangle$;
\STATE \hspace{30pt} {$\mathbf{u}^{*}_{t:m:k_{2},k_{1}}=\min\nolimits_{\alpha}\big\{c\left(\bm{s},\alpha\right)+\sum\nolimits_{m',k_{1}',k_{2}'}p\left(\bm{s}' |\bm{s},\alpha\right)\mathbf{u}^{*}_{t+1:m':k_{2}':k_{1}'}\cdot \bm{1}\{t\neq N\}\big\}$;}
\STATE  \hspace{30pt} {$\alpha_{\mathrm{opt}}=\arg\min\nolimits_{\alpha}\big\{c\left(\bm{s},\alpha\right)+\sum\nolimits_{m',k_{1}',k_{2}'}p\left(\bm{s}' |\bm{s},\alpha\right)\mathbf{u}^{*}_{t+1:m':k_{2}':k_{1}'}\cdot \bm{1}\{t\neq N\}\big\}$;}%,\\$\bm{\pi}_{t:m:k_{2}:k_{1}}=\alpha_{\mathrm{opt}}$;}
%\ENDIF
\STATE \hspace{30pt} \textbf{If} {$\alpha_{\mathrm{opt}}==1$} \textbf{then}
\STATE \hspace{40pt} {$\mathbf{u}^{*}_{t:m:k:k_{1}}=\mathbf{u}^{*}_{t:m:k_{2}:k_{1}},\forall k<k_{2}$;}
\STATE \hspace{40pt} {$\bm{\pi}_{t:m:k:k_{1}}=1,\forall k\leq k_{2}$;}
\STATE \hspace{40pt} {$k_{1}=k_{1}-1$;}
\STATE \hspace{30pt} \textbf{Else}
\STATE \hspace{40pt} {$\mathbf{u}^{*}_{t:m:k_{2}:k}=\mathbf{u}^{*}_{t:m:k_{2}:k_{1}},\forall k<k_{1}$;}
\STATE \hspace{40pt} {$k_{2}=k_{2}-1$;}
\STATE \hspace{30pt} \textbf{End if}
\STATE \hspace{20pt} \textbf{End while}
\STATE \hspace{10pt} \textbf{End for}
\STATE \hspace{10pt} {$t=t-1$;}
\STATE \textbf{End while}
\end{algorithmic}
\label{MBIAalgorithm}
\end{algorithm}

\subsection{Heuristic Online Policies}
Though the MBIA algorithm accelerates the conventional BIA, as the size of the state space grows, its execution time becomes unacceptable. Moreover, the memory requirement increases in order to store $\bm{\pi}^{*}$. To further reduce complexity, in this subsection, we propose a Look-Ahead policy and a Threshold-based Heuristic policy. First, we introduce a Greedy-Transmit policy as the performance benchmark.

1) \emph{Greedy-Transmit Policy:} The Greedy-Transmit policy always takes a higher priority to use the harvested energy. In each block, if the available harvested energy can afford to transmit the packet and the peak power constraint is not violated, the EH-BS will be assigned to serve the user, i.e., $
\alpha_{i}^{GT}=\bm{1}\{p_{H,i}^{inv}\leq \min\{\epsilon_{i}\slash \tau,p_{H}^{\max}\}\}, \forall i\in\mathcal{N}$.

2) \emph{Look-Ahead Policy:} The Look-Ahead policy is a simplified version of the discrete MDP approach, the idea of which has been previously adopted for online algorithms designs \cite{YMao14,IAhmed1403}. In the MDP approach, in the $i$-th block, we need to consider the randomness of the future $N-i$ blocks. However, in the Look-Ahead policy, we only consider
the randomness in the next block, which lowers the complexity at the expense of performance degradation. To obtain the Look-Ahead policy, we can simply plug in $N=2$ in the MBIA. During data transmission, if $i<N$, according to the system state, the action $\alpha_{i}^{LA}$ is determined by the corresponding entry in $\bm{\pi}^{LA}_{1,:,:,:}$. If $i=N$,  $\alpha_{i}^{LA}=\bm{1}\{p_{H,i}^{inv}\leq \min\{\epsilon_{i}\slash \tau,p_{H}^{\max}\}\}$.

3) \emph{Threshold-based Heuristic Policy:} The Greedy-Transmit policy  always tends to minimize the service cost in the current block, without utilizing the G-channel SI. On the other hand, the performance of the Look-Ahead policy is closely related to the state quantization levels, i.e., the computational complexity. Intuitions suggest that the EH-BS prefers to serve the user when the H-channel is in a good condition while the G-channel is in deep fading. Also, a higher battery energy state may further stimulate the EH-BS to transmit more packets. Thus, in the Threshold-based Heuristic policy, for $i<N$, if $p_{H,i}^{inv}\leq \min\{\epsilon_{i}\slash \tau,p_{H}^{\max}\}$, $\alpha_{i}^{TH}$ is determined by
\begin{equation}
\alpha^{TH}_{i}=
\bm{1}\big\{\epsilon_{i}\Xi\left(c_{i},p_{H,i}^{inv}\right)\geq \zeta \mathrm{P_{avg}}\tau \Xi\left(\lambda_{1},\lambda_{2}\right)\big\}.
\label{threshold}
\end{equation}
Otherwise, $\alpha_{i}^{TH}=0$.  For $i=N$, $\alpha_{i}^{TH}=\bm{1}\{p_{H,i}^{inv}\leq \min\{\epsilon_{i}\slash \tau,p_{H}^{\max}\}\}$.
In (\ref{threshold}), $\zeta$ is a scaling factor that can be obtained and optimized in prior, and $\lambda_{1}$, $\lambda_{2}$ are constants given by $\lambda_{1}=\mathbb{E}\{c_{i}\}$ and
$\lambda_{2}=\mathbb{E}\{p_{H,i}^{inv}|p_{H,i}^{inv}\leq p_{H}^{\max}\}$, respectively.
For Rayleigh fading channels, they can be simplified as
\begin{equation}
\lambda_{1}=w_{D}\left(1-e^{-A_{G}\slash\kappa}\right)+w_{G}\tau A_{G}\mathrm{E}_{1}\left(A_{G}\slash\kappa\right),
\end{equation}
\begin{equation}
\lambda_{2}=A_{H}\mathrm{E}_{1}\left(A_{H}\slash p_{H}^{\max}\right)\cdot e^{A_{H}\slash p_{H}^{\max}},
\label{lambda2}
\end{equation}
where $A_{j}=\mu_{j} \left(2^{\frac{R}{W\tau}}-1\right)\sigma^{2}g_{0}^{-1}d_{j}^{\theta}$, $\mu_{j}=\mathbb{E}\left[\gamma_{j,i}\right],j\in\{G,H\}$, and $\mathrm{E}_{1}\left(x\right)=\int_{x}^{+\infty}\frac{e^{-t}}{t}\cdot dt$. Indeed, this policy obeys Proposition \ref{propSecIVA1} and Proposition \ref{propSecIVA2}, i.e., it is a threshold-based policy.

\subsection{Possible Extensions}
In this subsection, possible extensions of the obtained results will be discussed. Although we focus on a single-user HES network in this paper, our results provide valuable guidelines for designing more general networks. Specifically, for HES networks where the EH-BS and the GP-BS coordinate to serve multiple users with OFDMA\footnote{Each user is assumed to be allocated with one sub-carrier with equal bandwidth. The system performance can be further optimized by adopting bandwidth allocation techniques \cite{DNg13,ZWangPartI,ZWangPartII}, which is beyond our scope.}, the GA algorithm can be directly applied for the non-causal SI scenario with complexity $\mathcal{O}\left(NK\left(N+K\right)\right)$, where $K$ is the number of users. {Since there is a transmit power constraint for all the users at each BS, the feasibility of the sum transmit power should be checked in addition to the energy causality constraint in each loop of the GA algorithm.} For the causal SI scenario, the MDP-based algorithms are computationally intractable due to high complexity. However, monotone structures of the optimal online solution still hold. Thus, the low-complexity Threshold-based Heuristic policy developed in Section IV-B can be modified for multiple users. Specifically, in the $i$-th block, $\alpha_{i}^{TH}$ (the user index is omitted for brevity) can be determined according to (\ref{threshold}) for each user, and the EH-BS serves the users with the largest values of $\Xi\left(c_{i},p_{H,i}^{inv}\right)$ and $\alpha_{i}^{TH}=1$ given the available energy and transmit power constraints. Sample simulation results for a two-user HES network will be provided in the Section V. For HES networks with multiple EH-BSs, the idea of central energy queue can be applied to form an alliance of the EH-BSs \cite{Thuc1412}. By adopting the proposed algorithms, the assigned BS for each user, i.e., either the GP-BS or one of the EH-BSs, can be determined.

{It is worthwhile to note that, as the network size grows, the amount of SI and the number of decision variables grow accordingly, and the coupling among different BSs/users in resource allocation becomes more severe, both of which make the optimization highly difficult. As a result, the extensions of the algorithms dedicated for single-user networks may not offer a complete solution, but they can still serve as a benchmark for future investigations.}

\section{Simulation Results}
In this section, we evaluate the performance of the proposed algorithms and show the unique grid energy consumption and QoS tradeoff in HES wireless networks. In the simulations, we set $w_{G}=1$, $R=50$ Kbits, $N=50$, $\tau = 1$ ms, $\sigma^{2} = -97.5$ dBm, $W = 10$ MHz,  $g_{0}=-40$ dB, $\theta = 4$, $p_{H}^{\max}=0.5$ W and $p_{G}^{\max}=2$ W. Without loss of generality, we assume $f_{E_{H}}\left(e\right)=1/E_{m},e\in\left[0,E_{m}\right]$ and $\gamma_{G,i},\gamma_{H,i}$ are exponentially distributed with mean $0$ dB. In the GA algorithm, $\Xi\left(p,q\right) = p/q$ is adopted. We use $K = 25$, $B_{m}=NE_{m}$ in the MBIA. The Greedy-Transmit policy is used as a baseline for the online setting, while exhaustive search is used to find the optimal solution of the off-line case for comparison.

\subsection{Total Service Cost Minimization}

\begin{figure}[h]
\begin{center}
    \label{TSCPavg}
   \includegraphics[width=0.48\textwidth]{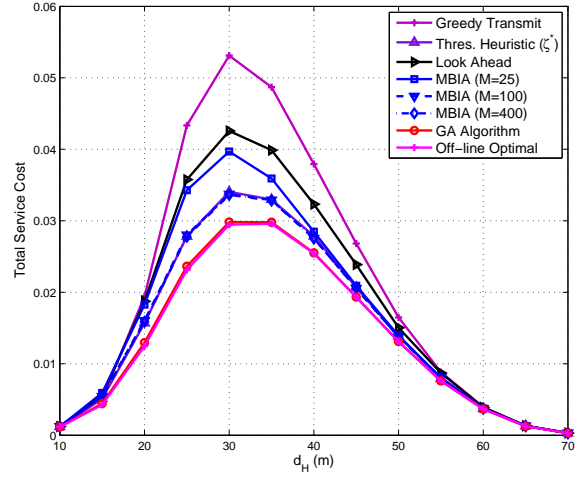}
\end{center}
\vspace{-15pt}
\caption{Total service cost vs. $d_{H}$, $w_{D}=0.01$, $\mathrm{P_{avg}}=20$ mW and $d_{H}+d_{G}=80$ m.}
\label{costdH}
\end{figure}

\begin{figure}[h]
\begin{center}
    \label{TSCPavg}
   \includegraphics[width=0.48\textwidth]{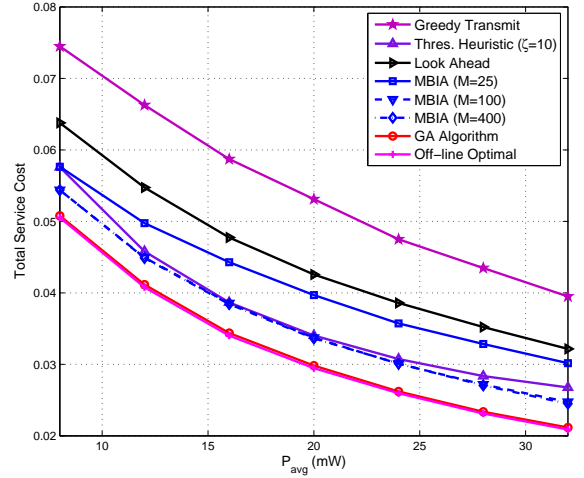}
\end{center}
\vspace{-15pt}
\caption{Total service cost vs. EH power, $w_{D}=0.01$, $d_{H}=30$ m, $d_{G}=50$ m.}
\label{costEHpwr}
\end{figure}

We start with the investigation of the impact of the distance from the EH-BS to the user, denoted as $d_{H}$, on the total service cost in Fig. \ref{costdH} by assuming that the user is located on the line between the two BSs. In the Threshold-based Heuristic policy, $\zeta^{*}$ is obtained via simulation within the candidate set $\left[0:0.5:200\right]$. We see that, when the user is near the EH-BS, the total service cost approaches zero, since with low path loss the EH-BS can deliver most of the data packets without incurring any cost. As the user moves closer to the GP-BS, we see that the total service cost decreases, which is due to the small path loss of the G-channel, and thus the GP-BS is able to transmit the packets with little grid energy consumption. In Fig. \ref{costEHpwr}, we fix $d_{H}=30$ m, $d_{G}=50$ m and show the total service cost versus the average EH power of different policies. In accordance with the intuition, the total service cost decreases as $\mathrm{P_{avg}}$ increases. As for performance comparison between different schemes, for the off-line case, from both Fig. \ref{costdH} and Fig. \ref{costEHpwr}, we see that the proposed GA algorithm enjoys near-optimal performance. For the online setting, the MBIA with $M = 25,100,400$ are evaluated\footnote{The number of possible states in the MDP is $1.25\times 10^{7}$ when $M = 400$. Thanks to MBIA, we can obtain the optimal policy within few hours.}, and the total service cost of $M = 100$ is significantly lower than that of $M = 25$. On the other hand, the performance gain from $M{=}100$ to $400$ is negligible, which indicates that $M=100$ is a viable choice. Moreover, all the proposed policies achieve a noticeable improvement compared to the benchmark  policy. Surprisingly, with a proper choice of $\zeta$, the proposed Threshold-based Heuristic policy not only outperforms the Look-Ahead policy and the MBIA with $M=25$, but also performs close to the MBIA with $M=100$ and $M=400$.
\vspace{-10pt}
\begin{figure}[h]
\begin{center}
    \label{TSCPavg}
   \includegraphics[width=0.48\textwidth]{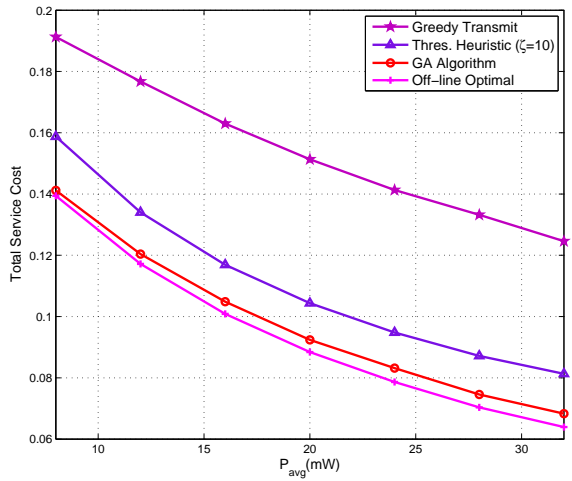}
\end{center}
\vspace{-15pt}
\caption{Total service cost vs. EH power (two users), $w_{D}=0.01$, both users are $30$ m from the EH-BS and $d_{G}=50$ m from the GP-BS.}
\label{costEHpwr2User}
\end{figure}

In Fig. \ref{costEHpwr2User}, we show the total service cost performance in a two-user HES network, where both users are $d_{H}=30$ m from the EH-BS and $d_{G}=50$ m from the GP-BS. It can be seen that the GA algorithm still achieves competitive performance compared to the off-line optimal algorithm. Also, the low-complexity Threshold-based Heuristic policy substantially reduces the total service cost compared to the benchmark policy, and its performance is comparable to those achieved by the off-line algorithms. These observations are similar with those from Fig. \ref{costEHpwr}, which indicates the potential of extending our findings to more general networks.

\subsection{Grid Energy Consumption and QoS Tradeoff}
\begin{figure}[h]
\begin{center}
    \label{TSCPavg}
   \includegraphics[width=0.48\textwidth]{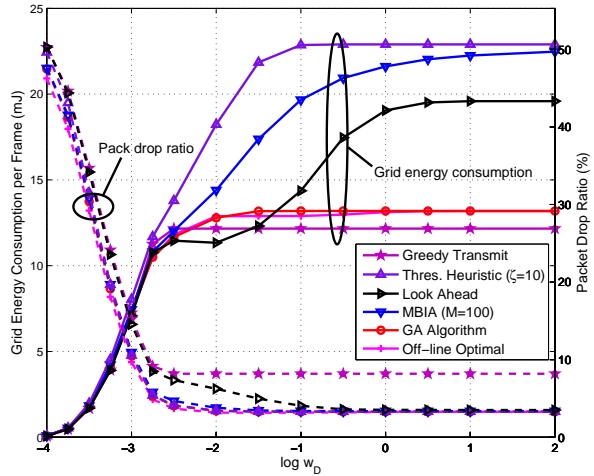}
\end{center}
\vspace{-15pt}
\caption{Grid energy consumption \& packet drop ratio vs. $w_{D}$, $d_{H}=30$ m, $d_{G}=50$ m, $\mathrm{P_{avg}}=20$ mW.}
\label{GridPacket}
\end{figure}
In this subsection, we will demonstrate the grid energy consumption and QoS tradeoff in HES wireless networks. We show the grid energy consumption and packet drop ratio, i.e., the percentage of the dropped packets, under different policies in Fig. \ref{GridPacket}. We see that, by adjusting $w_{D}$, different grid energy consumption and QoS tradeoffs are achieved. In particular, when $w_{D}$ is small, the QoS requirement is loose, so the network tends to save grid energy by dropping more packets. However, as $w_{D}$ gets larger, the grid energy consumption increases and the packet drop ratio decreases, which indicates that the network addresses more on QoS. Hence, given the QoS requirement, $w_{D}$ can be determined to minimize the grid energy consumption. For instance, assuming 96\% successful packet transmission is required, we may choose the Look-Ahead policy with $w_{D}=10^{-0.5}$ and the grid energy consumption will be around 17.5 mJ. We can also adopt the Threshold-based Heuristic policy with $w_{D}=10^{-2}$, which will enjoy a lower complexity, but consume slightly more grid energy, i.e., 18.2 mJ. In addition, the off-line policies are more competent to suppress both the grid energy consumption and packet drop ratio compared to the online policies. This is because of the availability of full SI that helps to identify the optimal H-blocks. Interestingly, when $w_{D}$ is relatively large, e.g. $>10^{-2}$, the grid energy consumed by the proposed online policies is significantly greater than that by the benchmark. The reason is that, with causal SI, the optimal policy should avoid the event of packet drop. Thus, in order to retain more harvested energy for the blocks with poor G-channel condition, the GP-BS will transmit more packets when the H-channel is in deep fading, which leads to higher grid energy consumption.
\vspace{-10pt}
\begin{figure}[h]
\begin{center}
    \label{TSCPavg}
   \includegraphics[width=0.48\textwidth]{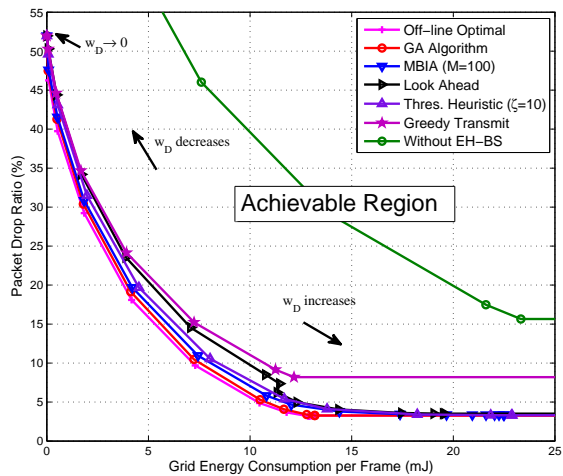}
\end{center}
\vspace{-15pt}
\caption{Packet drop ratio and grid energy consumption tradeoff regions.}
\label{tradeoffRegion}
\end{figure}

In Fig. \ref{tradeoffRegion}, we plot the grid energy consumption-QoS tradeoff region by adopting different values of $w_D$. The algorithm with a larger tradeoff region means it is more powerful to approach the Pareto-optimal performance achieved by the off-line optimal solution, i.e., the grid energy consumption (packet drop ratio) can not be further reduced without deteriorating the QoS (grid energy consumption) performance. We see the tradeoff region achieved by deploying a single GP-BS\footnote{In the case without an EH-BS, the Greedy Transmit policy is both off-line and online optimal.} is much smaller than those achieved by adding an EH-BS, i.e., with a given packet drop ratio requirement, the grid energy consumption is much higher and vice versa, which confirms the benefits of HES networks. Due to the limited allowable transmit power, there will always be a portion of dropped packets, e.g. when the EH-BS does not have enough energy while the G-channel is in deep fading, the packet has to be dropped. Consequently, we may observe non-zero floors of the packet drop ratio both in Fig. \ref{GridPacket} and Fig. \ref{tradeoffRegion}, which vary among different policies. For the Greedy Transmit policy, the packet drop ratio quickly saturates at 8.19\% when $w_{D}=10^{-1.5}$, which is due to its short-sighted nature, i.e., it always optimizes the current service cost. While for the Look-Ahead policy, the MBIA with $M=100$ and the Threshold-based Heuristic policy, the packet drop ratio keeps decreasing as $w_{D}$ increases and eventually reaches 3.51\%, 3.36\% and 3.32\%, respectively. This indicates that, in order to achieve a better QoS in HES networks, a certain amount of grid energy should be sacrificed, and at the same time, energy management schemes that fully exploit the available SI and properly balance the current/future system performance are needed.

\section{Conclusions}
In this paper, we investigated the grid energy consumption and QoS tradeoff in HES wireless networks. The total service cost addressing both the grid energy consumption and QoS was introduced as the performance metric, and BS assignment and power control policies were developed to optimize the system. It was revealed that the proposed total service cost minimization approach had the capability to achieve different grid energy consumption and QoS tradeoffs. Moreover, to maintain a better QoS in HES systems, apart from a higher grid energy consumption, the utilization of the harvested energy should consider both the current and future system performance based on the available side information. To exploit the full potential of HES networks in achieving green communications, further investigation will be needed. {It will be interesting to extend the BS assignment and power control policies to more general HES networks (e.g. both BSs are hybrid energy supplied).} Anther extension is to consider HES networks with wireless backhaul links, which will make it more attractive in practice. Other design problems should also be investigated, such as interference management, user scheduling, and channel estimation.

\begin{appendix}
\subsection{Proof for Lemma \ref{lmaSecIII1}}
First, it is straightforward to show that, if a packet is delivered by one of the BSs, without loss of optimality, (\ref{QoS}) can be achieved with equality, i.e., the assigned BS can transmit with the channel inversion power, which is the minimum power ensuring successful data packet transmission. Define $\alpha_{i}\triangleq I_{H,i}$, which indicates whether to deliver the packet with the EH-BS in the $i$-th block. Thus, the energy causality constraint and the peak power constraint for the EH-BS can be equivalently expressed as (\ref{EH causal constraint3}) and (\ref{H PPC2}), respectively. Moreover, given $\alpha_{i}=0$, when $p_{G,i}^{inv}\leq p_{G}^{\max}$ and the weighted cost of grid energy, $w_{G}p_{G,i}^{inv}\tau$, is less than the weighted cost of dropping a packet, $w_{D}$, the system should choose to transmit the packet, which incurs service cost $w_{G}p_{G,i}^{inv}\tau$. Otherwise, it means that the G-channel is experiencing deep fading, and thus the packet should be dropped, which incurs service cost $w_{D}$. Hence, the total service cost can be rewritten as $\sum_{i=1}^{N}\left(1-\alpha_{i}\right)c_{i}$, where $c_{i}$ is given by (\ref{costpara}). Besides, by contradiction, we can see that if $\{\alpha_{i}\}$ is optimal to $\hat{\mathcal{P}}_{1}$, (\ref{soluP3}) is optimal to $\mathcal{P}_{1}$.

\subsection{Proof for Property 3 of the GA algorithm}
For brevity, we will only prove the case when the H-channel is constant, i.e., $p_{H,i}^{inv}=p_{H}^{inv},\forall i \in\mathcal{N}$\footnote{We assume $p_{H}^{inv}\leq p_{H}^{\max}$, otherwise, there will be only one feasible solution $\alpha_{i}=0,\forall i\in\mathcal{N}$.}. The case when the G-channel is constant can be proved by using similar techniques. For convenience, we denote $\bm{\alpha}$ and $\tilde{\bm{\alpha}}$ as solutions obtained by the GA algorithm and any other feasible solution, respectively. Denote $|\bm{\alpha}|$ as the number of H-blocks in solution $\bm{\alpha}$, $m_{k}$ as the $k$-th ($k\leq |\alpha|$) block being selected as an H-block in the GA algorithm. $\left[k\right]$ ($k\leq |\tilde{\bm{\alpha}}|$) is the H-block index in solution $\tilde{\bm{\alpha}}$ with the $k$-th largest value of $\Xi\left[c_{i},p_{H,i}^{inv}\right]$. Further, $\mathcal{N}_{\bm{\alpha}}\triangleq\{i|\alpha_{i}=1\}$. We will show the total service cost $c_{\Sigma}^{\tilde{\bm{\alpha}}}$ achieved by $\forall \tilde{\bm{\alpha}}\neq \bm{\alpha}$ will be no smaller than $c_{\Sigma}^{\mathrm{GA}}$.

Owing to the greedy manner of the GA algorithm, $c_{m_{1}}\geq c_{m_{2}}\cdots \geq c_{m_{|\bm{\alpha}|}}$. According to the definition of $\left[k\right]$, we have $c_{\left[1\right]}\geq c_{\left[2\right]} \cdots \geq c_{\left[|\tilde{\bm{\alpha}}|\right]}$. Define
\begin{equation}
\begin{split}
&i_{0}=\mathrm{dist}\left(\bm{\alpha},\tilde{\bm{\alpha}}\right)\triangleq\\
&\begin{cases}
\min\{|\bm{\alpha}|,|\tilde{\bm{\alpha}}|\}+1, & m_{i}=\left[i\right], \forall i=1,\cdots \min\{|\bm{\alpha}|,|\tilde{\bm{\alpha}}|\}\\
\min\{i|\mathrm{s.t.}\ m_{i}\neq \left[i\right]\}, &{\rm{otherwise}}.
\end{cases}
\end{split}
\end{equation}
If $\mathrm{dist}\left(\bm{\alpha},\tilde{\bm{\alpha}}\right)=\min\{|\bm{\alpha}|,|\tilde{\bm{\alpha}}|\}+1$, because of the greedy manner, we have $|\bm{\alpha}|\geq |\tilde{\bm{\alpha}}|$, which leads to $c_{\Sigma}^{\tilde{\bm{\alpha}}}\geq c_{\Sigma}^{\mathrm{GA}}$. Now, we consider the case that $\mathrm{dist}\left(\bm{\alpha},\tilde{\bm{\alpha}}\right)\leq \min\{|\bm{\alpha}|,|\tilde{\bm{\alpha}}|\}$. Denote
\begin{equation}
\begin{split}
&\mathcal{N}_{\tilde{\bm{\alpha}}}^{L,i_{0}}\triangleq\{i|i<m_{i_{0}},\tilde{\alpha}_{i}=1,\forall k<i_{0},i\neq \left[k\right]\},\\
&\mathcal{N}_{\tilde{\bm{\alpha}}}^{R,i_{0}}\triangleq\{i|i>m_{i_{0}},\tilde{\alpha}_{i}=1,\forall k<i_{0},i\neq \left[k\right]\}.
\end{split}
\end{equation}
Thus, $\left[i_{0}\right]\in \mathcal{N}_{\tilde{\bm{\alpha}}}^{L,i_{0}}\cup\mathcal{N}_{\tilde{\bm{\alpha}}}^{R,i_{0}}$ and $c_{m_{i_{0}}}\geq c_{i}, \forall i\in \mathcal{N}_{\tilde{\bm{\alpha}}}^{L,i_{0}}\cup\mathcal{N}_{\tilde{\bm{\alpha}}}^{R,i_{0}}$.

If $\mathcal{N}_{\tilde{\bm{\alpha}}}^{L,i_{0}}\neq\emptyset$, by choosing an arbitrary $\overline{i}_{0}\in \mathcal{N}_{\tilde{\bm{\alpha}}}^{L,i_{0}}$, we can construct a new solution
\begin{equation}
\{\tilde{\alpha}^{1}_{i}\}=
\begin{cases}
1, &i=m_{i_{0}}\\
0, &i=\overline{i}_{0}\\
\tilde{\alpha}_{i}, &{\rm{otherwise}}
\end{cases}
\label{newfeassolution}
\end{equation}
with $\mathrm{dist}\left(\bm{\alpha},\tilde{\bm{\alpha}}^{1}\right)=\mathrm{dist}\left(\bm{\alpha},
\tilde{\bm{\alpha}}\right)+1$ and $c_{\Sigma}^{\tilde{\bm{\alpha}}^{1}}\leq c_{\Sigma}^{\tilde{\bm{\alpha}}}$. If $\mathcal{N}_{\tilde{\bm{\alpha}}}^{L,i_{0}}=\emptyset$, but $\mathcal{N}_{\tilde{\bm{\alpha}}}^{R,i_{0}}\neq \emptyset$, we can choose $\overline{i}_{0}=\min\{i|i\in \mathcal{N}_{\tilde{\bm{\alpha}}}^{R,i_{0}}\}$ and construct a new feasible $\tilde{\bm{\alpha}}^{1}$ according to (\ref{newfeassolution}) (We will show its feasibility at the end  of this proof), with $\mathrm{dist}\left(\bm{\alpha},\tilde{\bm{\alpha}}^{1}\right)=\mathrm{dist}\left(\bm{\alpha},
\tilde{\bm{\alpha}}\right)+1$ and $c_{\Sigma}^{\tilde{\bm{\alpha}}^{1}}\leq c_{\Sigma}^{\tilde{\bm{\alpha}}}$.

By repeating the above steps for at most $\min\{|\bm{\alpha}|,|\tilde{\bm{\alpha}}|\}+1-i_{0}$ times, we will be able to conclude $c_{\Sigma}^{\tilde{\bm{\alpha}}}\geq c_{\Sigma}^{\tilde{\bm{\alpha}}^{1}}\geq \cdots \geq c_{\Sigma}^{\mathrm{GA}}$. Thus, $\bm{\alpha}$ is optimal.

It remains to show the feasibility as mentioned above. Since $\tilde{\bm{\alpha}}$ is feasible, the energy causality constraints are satisfied, i.e.,
\begin{equation}
\begin{split}
\sum_{i=1}^{k}\tilde{\alpha}_{i}p_{H}^{inv}\tau&=\sum_{i=1}^{k} p_{H}^{inv}\tau \cdot \bm{1}\{i\in \mathcal{M}_{i_{0}-1}\cup \mathcal{N}_{\tilde{\bm{\alpha}}}^{R,i_{0}}\}\\
&\leq \sum_{l=1}^{k}E_{H,l},\forall k\in\mathcal{N},
\label{feasibility1}
\end{split}
\end{equation}
where $\mathcal{M}_{i_{0}-1}\triangleq\{\left[1\right],\cdots,\left[i_{0}-1\right]\}=\{m_{1},\cdots,m_{i_{0}-1}\}$. Since $\mathcal{M}_{i_{0}}\subseteq \mathcal{N}_{\bm{\alpha}}$, thus, $\hat{\bm{\alpha}}\triangleq \{\hat{\alpha}_{i}=1,\forall i\in\mathcal{M}_{i_{0}}\}$ is feasible, i.e.,
\begin{equation}
\sum_{i=1}^{k}\hat{\alpha}_{i}p_{H}^{inv}\tau=\sum_{i=1}^{k}p_{H}^{inv}\tau \cdot \bm{1}\{i\in\mathcal{M}_{i_{0}}\}\leq\sum_{l=1}^{k}E_{H,l},\forall k\in\mathcal{N}.
\label{feasibility2}
\end{equation}
Suppose that $\tilde{\bm{\alpha}}^{1}$ is infeasible, i.e., $\exists k_{0}\in \mathcal{N}$, such that
\begin{equation}
\begin{split}
\sum_{i=1}^{k_{0}}\tilde{\alpha}^{1}_{i}p_{H}^{inv}\tau &=\sum_{i=1}^{k_{0}}p_{H}^{inv}\tau \cdot \bm{1} \{i\in \mathcal{M}_{i_{0}}\cup\left(\mathcal{N}_{\tilde{\bm{\alpha}}}^{R,i_{0}}\setminus\{\overline{i}_{0}\}\right)\}\\
&>\sum_{l=1}^{k_{0}}E_{H,l}.
\label{feasibility3}
\end{split}
\end{equation}
If $1\leq k_{0}\leq \overline{i}_{0}-1$, since $\overline{i}_{0}=\min\{i|i\in\mathcal{N}_{\tilde{\bm{\alpha}}}^{R,i_{0}}\}$, the left hand side (LHS) of (\ref{feasibility3})
equals the LHS of (\ref{feasibility2}) when $k=k_{0}$, which violates the assumption that $\hat{\bm{\alpha}}$ is feasible. If $\overline{i}_{0}\leq k_{0}\leq N$, the LHS of (\ref{feasibility3}) equals the LHS of (\ref{feasibility1}) when $k=k_{0}$, which contradicts the feasibility of $\tilde{\bm{\alpha}}$. Hence, $\tilde{\bm{\alpha}}^{1}$ is feasible.

\subsection{Proof for Lemma \ref{lmaSecIVA2}}
We conduct the proof with mathematical induction. Denote $\epsilon^{+}\geq \epsilon^{-}$, $\bm{s}_{i}^{+}=\langle \epsilon^{+},\gamma_{G,i},\gamma_{H,i} \rangle$ and $\bm{s}_{i}^{-}=\langle \epsilon^{-},\gamma_{G,i},\gamma_{H,i} \rangle$. For $i=N$,
\begin{equation}
\begin{split}
\hat{u}_{N}^{*}\left(\epsilon^{+}\right)&=\sum_{\gamma_{G,N},\gamma_{H,N}}u_{N}^{*}
\left(\bm{s}^{+}_{N}\right)\\ &=\sum_{\gamma_{G,N},\gamma_{H,N}}\min_{\alpha^{+}_{N}\in\mathcal{A}_{\bm{s}^{+}_{N}}}c\left(\bm{s}^{+}_{N},\alpha_{N}^{+}\right),
\end{split}
\end{equation}
\begin{equation}
\begin{split}
\hat{u}_{N}^{*}\left(\epsilon^{-}\right)&=\sum_{\gamma_{G,N},\gamma_{H,N}}u_{N}^{*}
\left(\bm{s}^{-}_{N}\right) \\&=\sum_{\gamma_{G,N},\gamma_{H,N}}\min_{\alpha^{-}_{N}\in\mathcal{A}_{\bm{s}^{-}_{N}}}c\left(\bm{s}^{-}_{N},\alpha_{N}^{-}\right).
\end{split}
\end{equation}
It can be easily checked that $\forall\gamma_{G,N},\gamma_{H,N}$,
\begin{equation}
\min_{\alpha^{+}_{N}\in\mathcal{A}_{\bm{s}^{+}_{N}}}c\left(\bm{s}^{+}_{N},\alpha_{N}^{+}\right)\leq\min_{\alpha^{-}_{N}\in\mathcal{A}_{\bm{s}^{-}_{N}}}c\left(\bm{s}^{-}_{N},\alpha_{N}^{-}\right).
\end{equation}
Thus, $\hat{u}_{N}^{*}\left(\epsilon^{+}\right)\leq \hat{u}_{N}^{*}\left(\epsilon^{-}\right)$.
Suppose for $1<n< N$, $\hat{u}_{n}^{*}\left(\epsilon^{+}\right)\leq \hat{u}_{n}^{*}\left(\epsilon^{-}\right)$,
\begin{equation}
\begin{split}
&u_{n-1}^{*}\left(\bm{s}^{+}_{n-1}\right)=\\
&c\left(\bm{s}_{n-1}^{+},\alpha_{n-1}^{+*}\right)+K^{-2}\sum_{\epsilon^{+'}}p\left(\epsilon^{+'}|\bm{s}_{n-1}^{+},\alpha_{n-1}^{+*}\right)\hat{u}_{n}^{*}\left( \epsilon^{+'}\right),
\end{split}
\end{equation}
\begin{equation}
\begin{split}
&u_{n-1}^{*}\left(\bm{s}^{-}_{n-1}\right)=\\
&c\left(\bm{s}_{n-1}^{-},\alpha_{n-1}^{-*}\right)+K^{-2}\sum_{\epsilon^{-'}}p\left(\epsilon^{-'}|\bm{s}_{n-1}^{-},\alpha_{n-1}^{-*}\right)\hat{u}_{n}^{*}\left( \epsilon^{-'}\right),
\end{split}
\end{equation}
where $\alpha_{n-1}^{+*}$ and $\alpha_{n-1}^{-*}$ is the optimal action in the $\left( n-1 \right)$-th block given states $\bm{s}_{n-1}^{+}$ and $\bm{s}_{n-1}^{-}$, respectively. $\epsilon^{+'}$ ($\epsilon^{-'}$ ) is the possible energy state in the next block if $\alpha_{n-1}^{+*}$ ($\alpha_{n-1}^{-*}$) is taken.
We will show that under any combination of  $\langle \alpha_{n-1}^{+*} ,\alpha_{n-1}^{-*} \rangle$, i.e., $\langle 0, 0\rangle$, $\langle 0, 1\rangle$, $\langle 1, 0\rangle$,$\langle 1, 1\rangle$, $u_{n-1}^{*}\left(\bm{s}^{+}_{n-1}\right)\leq u_{n-1}^{*}\left(\bm{s}^{-}_{n-1}\right)  $.

We provide the following inequality to facilitate the proof, which can be verified by similar procedures as those in the proof of Lemma \ref{lmaSecIVA1},
\begin{equation}
\begin{split}
&\sum_{\epsilon^{+'}}p\left(\epsilon^{+'}|\bm{s}_{n-1}^{+},\alpha\right)\hat{u}_{n}^{*}\left( \epsilon^{+'}\right)\\
&\leq\sum_{\epsilon^{-'}}p\left(\epsilon^{-'}|\bm{s}_{n-1}^{-},\alpha\right)\hat{u}_{n}^{*}\left( \epsilon^{-'}\right).
\label{AssistIneq}
\end{split}
\end{equation}
For $\langle \alpha_{n-1}^{+*} ,\alpha_{n-1}^{-*} \rangle=\langle 0, 0\rangle$,
\begin{equation}
\begin{split}
&u^{*}_{n-1}\left(\bm{s}_{n-1}^{+}\right)\\&=c\left(\bm{s}_{n-1}^{+},0\right)+K^{-2}\sum_{\epsilon^{+'}}p\left(\epsilon^{+'}|\bm{s}_{n-1}^{+},0\right)\hat{u}_{n}^{*}\left( \epsilon^{+}\right)\\
&\overset{\left(b\right)}{\leq} c\left(\bm{s}_{n-1}^{-},0\right) +K^{-2}\sum_{\epsilon^{-'}}p\left(\epsilon^{-'}|\bm{s}_{n-1}^{-},0\right)\hat{u}_{n}^{*}\left( \epsilon^{-}\right)\\&=u^{*}_{n-1}\left(\bm{s}_{n-1}^{-}\right).
\end{split}
\end{equation}
For $\langle \alpha_{n-1}^{+*} ,\alpha_{n-1}^{-*} \rangle=\langle 0, 1\rangle$,
\begin{equation}
\begin{split}
&u^{*}_{n-1}\left(\bm{s}_{n-1}^{+}\right)\\&\overset{\left(c\right)}{\leq} c\left(\bm{s}_{n-1}^{+},1\right)+K^{-2}\sum_{\epsilon^{+'}}p\left(\epsilon^{+'}|\bm{s}_{n-1}^{+},1\right)\hat{u}_{n}^{*}\left( \epsilon^{+}\right)\\
&\overset{\left(b\right)}{\leq}  c\left(\bm{s}_{n-1}^{-},1\right)+K^{-2}\sum_{\epsilon^{-'}}p\left(\epsilon^{-'}|\bm{s}_{n-1}^{-},1\right)\hat{u}_{n}^{*}\left( \epsilon^{-}\right)\\&=u^{*}_{n-1}\left(\bm{s}_{n-1}^{-}\right).
\end{split}
\end{equation}
For $\langle \alpha_{n-1}^{+*} ,\alpha_{n-1}^{-*} \rangle=\langle 1, 0\rangle$,
\begin{equation}
\begin{split}
&u^{*}_{n-1}\left(\bm{s}_{n-1}^{+}\right)=K^{-2}\sum_{\epsilon^{+'}}p\left(\epsilon^{+'}|\bm{s}_{n-1}^{+},1\right)\hat{u}_{n}^{*}\left( \epsilon^{+}\right)\\
&\overset{\left(c\right)}{\leq} c\left(\bm{s}_{n-1}^{+},0\right)+  K^{-2}\sum_{\epsilon^{+'}}p\left(\epsilon^{+'}|\bm{s}_{n-1}^{+},0\right)\hat{u}_{n}^{*}\left( \epsilon^{+}\right)\\
&\overset{\left(b\right)}{\leq} c\left(\bm{s}_{n-1}^{-},0\right)+  K^{-2}\sum_{\epsilon^{-'}}p\left(\epsilon^{-'}|\bm{s}_{n-1}^{-},0\right)\hat{u}_{n}^{*}\left( \epsilon^{-}\right)\\&=u^{*}_{n-1}\left(\bm{s}_{n-1}^{-}\right).
\end{split}
\end{equation}
For $\langle \alpha_{n-1}^{+*} ,\alpha_{n-1}^{-*} \rangle=\langle 1, 1\rangle$,
\begin{equation}
\begin{split}
&u^{*}_{n-1}\left(\bm{s}_{n-1}^{+}\right)=K^{-2}\sum_{\epsilon^{+'}}p\left(\epsilon^{+'}|\bm{s}_{n-1}^{+},1\right)\hat{u}_{n}^{*}\left( \epsilon^{+}\right)\\
&\overset{\left(b\right)}{\leq} K^{-2}\sum_{\epsilon^{-'}}p\left(\epsilon^{-'}|\bm{s}_{n-1}^{-},1\right)\hat{u}_{n}^{*}\left( \epsilon^{-}\right)\\&=u^{*}_{n-1}\left(\bm{s}_{n-1}^{-}\right).
\end{split}
\end{equation}
In the above equations, $\left(b\right)$ holds because of (\ref{AssistIneq}), and $\left(c\right)$ holds due to (\ref{Bellman}). As a result, $\forall \gamma_{G,n-1},\gamma_{H,n-1}$, $u_{n-1}^{*}\left(\bm{s}^{+}_{n-1}\right)\leq u_{n-1}^{*}\left(\bm{s}^{-}_{n-1}\right) $. By summing up both sides of the inequality over all $\gamma_{G,n-1},\gamma_{H,n-1}$, we have $\hat{u}_{n-1}^{*}\left(\epsilon^{+}\right)\leq \hat{u}_{n-1}^{*}\left(\epsilon^{-}\right)$. Consequently, $1\leq n \leq N$, $\hat{u}_{n}^{*}\left(\epsilon^{+}\right)\leq \hat{u}_{n}^{*}\left(\epsilon^{-}\right)$.
\end{appendix}

%\section*{Acknowledgment}


\begin{thebibliography}{1}

\bibitem{Fehske11}
A. Fehske, G. Fettweis, J. Malmodin, and G. Biczok, ``The global footprint of mobile communications: The ecological and economic perspective,'' \emph{IEEE Commun. Mag.}, vol. 49, no. 8, pp. 55-62, Aug. 2011.

\bibitem{Lambert12}
S. Lambert \emph{et al.}, ``World-wide electricity consumption of communication networks,'' \emph{Optical Express}, vol. 20, no. 26, pp. B513-524, Mar. 2012.

\bibitem{Sudevalayam11}
S. Sudevalayam and P. Kulkarni, ``Energy harvesting sensor nodes: Survey
and implications,'' \emph{IEEE Commun. Surveys Tuts.}, vol. 13, no. 3, pp. 443-461, Jul. 2011.

\bibitem{Ulukus15}
S. Ulukus, A. Yener, E. Erkip, O. Simeone, M. Zorzi, and K. Huang, ``Energy harvesting wireless communications: A review of recent advances,'' \emph{IEEE J. Sel. Areas Commun.}, vol. 33, no.3, pp. 360-381, Mar. 2015.

\bibitem{Piro13}
G. Piro \emph{et al.}, ``Hetnets powered by renewable energy sources: Sustainable next generation cellular networks,'' \emph{IEEE Internet Comput.}, vol. 17, no. 1, pp. 32-39, Jan. 2013.

\bibitem{Ozel11}
O. Ozel, K. Tutuncuoglu, J. Yang, S. Ulukus, and A. Yener, ``Transmission with energy harvesting nodes in fading wireless channels: Optimal policies,'' \emph{IEEE J. Sel. Areas Commun.}, vol. 29, no. 8, pp. 1732-1743, Sep. 2011.

\bibitem{ZWang14}
Z. Wang, V. Aggarwal, and X. Wang, ``Power allocation for energy harvesting transmitter with causal information,'' \emph{IEEE Trans. Commun.}, vol. 62, no. 11, pp. 4080-4093, Nov. 2014.

\bibitem{Blasco13}
P. Blasco, D. Gunduz, and M. Dohler, ``A learning theoretic approach to energy harvesting communication system optimization,'' \emph{IEEE Trans. Wireless Commun.}, vol. 12, no. 4, pp. 1872-1882, Apr. 2013.

\bibitem{YLuo1212}
Y. Luo, J. Zhang, and K. B. Letaief, ``Training optimization for energy harvesting communication systems,'' in \emph{Proc. IEEE Global Commun. Conf. (GLOBECOM)}, Anaheim, CA, Dec. 2012, pp. 3365-3370.

\bibitem{JYang12}
J. Yang, O. Ozel, and S. Ulukus, ``Broadcasting with an energy harvesting rechargeable transmitter,'' \emph{IEEE Trans. Wireless Commun.}, vol. 11, no. 2, pp. 571-583, Feb. 2012.

\bibitem{JYang122}
J. Yang and S. Ulukus, ``Optimal packet scheduling in a multiple access channel with energy harvesting transmitters,'' \emph{J. Commun. \& Networks}, vol. 14, no. 2, pp. 140-150, Apr. 2012.

\bibitem{CHuang13}
C. Huang, R. Zhang, and S. Cui, ``Throughput maximization for the Gaussian relay channel with energy harvesting constraints,'' \emph{IEEE J. Sel. Areas Commun.}, vol. 31, no. 8, pp. 1469-1479, Aug. 2013.

\bibitem{YLuo13}
Y. Luo, J. Zhang, and K. B. Letaief, ``Optimal scheduling and power allocation for two-hop energy harvesting communication systems,'' \emph{IEEE Trans. Wireless Commun.}, vol. 11, no. 9, pp. 4729-4741, Sep. 2013.

\bibitem{YLuo1312}
\raisebox{0.5mm}{------}, ``Relay selection for energy harvesting cooperative communication systems,'' in \emph{Proc. IEEE Global Commun. Conf. (GLOBECOM)}, Atlanta, GA, Dec. 2013, pp. 2514-2519.

\bibitem{YLuo1601}
\raisebox{0.5mm}{------}, ``Transmit power minimization for wireless networks with energy harvesting relays,'' \emph{IEEE Trans. Commun.},
to appear, available at http://arxiv.org/abs/1601.04247.

\bibitem{YMao14}
Y. Mao, J. Zhang, S. H. Song, and K. B. Letaief, ``Joint link selection and relay power allocation for energy harvesting relaying systems,'' in \emph{Proc. IEEE Global Commun. Conf. (GLOBECOM)}, Austin, TX, Dec. 2014, pp. 2568-2573.

\bibitem{ZWangPartI}
Z. Wang, V. Aggarwal, and X. Wang, ``Multiuser joint energy-bandwidth allocation with energy  harvesting - part I: Optimum algorithm \& multiple point-to-point channels,'' available at http://arxiv.org/abs/1410.2861.

\bibitem{ZWangPartII}
\raisebox{0.5mm}{------}, ``Multiuser joint energy-bandwidth allocation with energy  harvesting - part II: Multiple
broadcast channels \& proportional fairness,'' available at http://arxiv.org/abs/1410.2867.

\bibitem{YMao1312}
Y. Mao, G. Yu, and C. Zhong, ``Energy consumption analysis of energy harvesting systems with power grid,'' \emph{IEEE Wireless Commun. Lett.}, vol. 2, no. 6, pp. 611-614, Dec. 2013.

\bibitem{JGong13}
J. Gong, S. Zhou, and Z. Niu, ``Optimal power allocation for energy harvesting and power grid coexisting wireless communication systems,'' \emph{IEEE Trans. Commun.}, vol. 61, no. 7, pp. 3040-3049, Jul. 2013.

\bibitem{XKang1408}
X. Kang, Y. -K. Chia, C. -K. Ho, and S. Sun, ``Cost minimization for fading channels with energy harvesting and conventional energy,'' \emph{IEEE Trans. Wireless Commun.}, vol. 13, no. 8, pp. 4586-4598, Aug. 2014.

\bibitem{DNg13}
D. Ng, E. S. Lo, and R. Schober, ``Energy-efficient resource allocation in OFDMA systems with hybrid energy harvesting base station,'' \emph{IEEE Trans. Wireless Commun.}, vol. 12, no. 7, pp. 3412-3427, Jul. 2013.

\bibitem{THan1302}
T. Han and N. Ansari, ``On optimizing green energy utilization for
cellular networks with hybrid energy supplies,'' \emph{IEEE Trans. Wireless Commun.}, vol. 12, no. 8, pp. 3872-3882, Aug. 2013.

\bibitem{JGong14}
J. Gong, J. S. Thompson, S. Zhou, and Z. Niu, ``Base station sleeping and resource allocation in renewable energy powered cellular networks,'' \emph{IEEE Trans. Commun.}, vol. 62, no. 11, pp. 3801-3813, Nov. 2014.

\bibitem{THan13}
T. Han and N. Ansari, ``Green-energy aware and latency aware user association in heterogeneous cellular networks,'' in \emph{Proc. IEEE Global Commun. Conf. (GLOBECOM)}, Atlanta, GA, USA, Dec. 2013, pp. 4946-4951.

\bibitem{JXu14}
J. Xu and R. Zhang, ``CoMP meets smart grid: A new communication and energy cooperation paradigm,'' \emph{IEEE Trans. Veh. Technol.}, vol. 64, no. 6, pp. 2476-2488, Jun. 2015.

\bibitem{YMao1512}
Y. Mao, J. Zhang, and K. B. Letaief, ``A Lyapunov optimization approach for green cellular networks with hybrid energy supplies,'' \emph{IEEE J. Sel. Areas Commun.}, vol. 33, no. 12, pp. 2463-2477, Dec. 2015.

\bibitem{YMao1501}
\raisebox{0.5mm}{------}, ``Joint base station assignment and power control for hybrid energy supply wireless networks,'' in \emph{Proc. IEEE Wireless Commun. Networking Conf. (WCNC)}, New Orleans, LA, Mar. 2015, pp. 1177-1182.

\bibitem{YMao15}
Y. Mao, Y. Luo, J. Zhang, and K. B. Letaief, ``Energy harvesting small cell networks: Feasibility, deployment and operation,'' \emph{IEEE Commun. Mag.}, vol. 53, no. 6, pp. 94-101, Jun. 2015.

\bibitem{CLi14}
C. Li, J. Zhang, and K. B. Letaief, ``Throughput and energy efficiency analysis of small cell networks with multi-antenna base stations,'' \emph{IEEE Trans. Wireless Commun.}, vol. 13, no. 5, pp. 2502-2517, May. 2014.

\bibitem{Jeff11}
J. G. Andrews, F. Baccelli, and R. K. Ganti, ``A tractable approach to coverage and rate in cellular networks,'' \emph{IEEE Trans. Commun.}, vol. 59, no. 11, pp. 3122-3134, Nov. 2011.

\bibitem{JZhang0904}
J. Zhang, R. Chen, J. G. Andrews, A. Ghosh, and R. W. Heath Jr., ``Networked MIMO with clustered linear precoding,'' \emph{IEEE Trans. Wireless Commun.}, vol. 8, no. 4, pp. 1910-1921, Apr. 2009.

\bibitem{Biermann13}
T. Biermann, L. Scalia, C. Choi, W. Kellerer, and H. Karl, ``How backhaul networks influence the feasibility of coordinated multipoint in cellular networks,'' \emph{IEEE Commun. Mag.}, vol. 51, no. 8, pp. 168-176, Aug. 2013.

\bibitem{QCui14}
Q. Cui \emph{et al.}, ``Evolution of limited-feedback CoMP systems from 4G to 5G: CoMP features and limited-feedback approaches,'' \emph{IEEE Veh. Technol. Mag.}, vol. 9, no. 3, pp. 94-103, Sep. 2014.

\bibitem{AGiovanidis09}
A. Giovanidis, G. Wunder, and J. Buhler, ``Optimal control of a single queue with retransmissions: Delay-dropping tradeoffs,'' \emph{IEEE Trans. Wireless Commun.}, vol. 8, no. 7, pp. 3736-3746, Jul. 2009.

\bibitem{IAhmed1312}
I. Ahmed, A. Ikhlef, R. Schober, D. Ng, and R. K. Mallik, ``Power allocation for an energy harvesting transmitter with hybrid energy sources,'' \emph{IEEE Trans. Wireless Commun.}, vol. 12, no. 12, pp. 6255-6267, Dec. 2013.

\bibitem{multiKnapsack}
A. Freville, ``The multidimensional 0-1 knapsack problem: An overview,'' \emph{Eur. J. Oper. Res.}, vol. 155, no. 1, pp. 1-21, May. 2004.

\bibitem{KellerKnapsack}
H. Kellerer, U. Pferschy, and D. Pisinger, \emph{Knapsack problem}. Heidelber, Berlin, Germany: Springer, 2004.

\bibitem{MOSEK}
``MOSEK Software,'' Available: http://www.mosek.com.

\bibitem{HSWang95}
H. S. Wang and N. Moayeri, ``Finite-state Markov channel - a useful model for radio communication channels,'' \emph{IEEE Trans. Veh. Technol.}, vol. 44, no.1, pp. 163-171, Feb. 1995.

\bibitem{BertsekasDP}
D. P. Bertsekas, \emph{Dynamic programming and optimal control} 3rd ed.
Belmonth, MA, USA: Athens Scientific, 2005.

\bibitem{IAhmed1403}
I. Ahmed, A. Ikhlef, R. Schober, and R. K. Mallik, ``Power allocation for conventional and buffer-aided link adaptive relaying systems with energy harvesting nodes,'' \emph{IEEE Trans. Wireless Commun.}, vol. 13, no. 3, pp. 1182-1195, Mar. 2014.

\bibitem{BBai1009}
B. Bai, W. Chen, Z. Cao, and K. B. Letaief, ``Uplink cross-layer scheduling with differential QoS requirements in OFDMA systems,'' \emph{EURASIP Journal on Wireless Communication and Networking}, pp. 1-10, Sep. 2010.

\bibitem{Thuc1412}
T. Thuc, H. Tabassum, and E. Hossain, ``A stochastic power control game for two-tier cellular networks with energy harvesting small cells,'' in \emph{Proc. IEEE Global Commun. Conf. (GLOBECOM)}, Austin, TX, Dec. 2014, pp. 2637-2642.
\end{thebibliography}
\end{document}